\RequirePackage{fix-cm}
\documentclass[pdflatex,iicol,sn-mathphys-num]{sn-jnl}

\usepackage{graphicx}
\usepackage{amsmath,mathtools,amssymb,amsfonts}
\usepackage{empheq}
\usepackage{subfigure}
\usepackage{tikz}
\usepackage{tikz-cd}
\usetikzlibrary{decorations.markings}
\usetikzlibrary{shapes.geometric, arrows}
\usepackage{tikzscale}
\usepackage[integrals]{wasysym}
\DeclareFontFamily{U}{wasy}{}
\DeclareFontShape{U}{wasy}{m}{n}{<->wasy10}{}
\DeclareFontShape{U}{wasy}{b}{n}{<->wasyb10}{}
\usepackage[dvipsnames]{xcolor}
\usepackage{changes}
\usepackage{placeins}
\usepackage{tcolorbox}
\usepackage{float}
\usepackage[scaled=.95]{helvet}
\usepackage{bbm}
\usepackage{amscd}
\usepackage{commath}
\usepackage{lmodern}
\usepackage{bm}
\usepackage{physics}
\usepackage{enumitem}
\usepackage{booktabs}
\usepackage{tabularx}
\usepackage{array}
\usepackage{multirow}
\usepackage{comment}
\usepackage[english]{babel}
\usepackage[IL2]{fontenc}
\usepackage[utf8]{inputenc}
\usepackage[linesnumbered,ruled,vlined]{algorithm2e}
\usepackage[noend]{algpseudocode}
\usepackage{ulem}
\usepackage{cleveref}

\hypersetup{colorlinks=true, urlcolor=blue, linkcolor=blue,
citecolor=blue1, hypertexnames=false}

\definecolor{blue1}{rgb}{0.0, 0.0, 1.0}
\definecolor{gray}{rgb}{0.9,0.9,0.9}
\definecolor{gray1}{rgb}{0.7,0.7,0.7}
\definecolor{gray2}{rgb}{0.8,0.8,0.8}
\definecolor{magenta}{rgb}{1.0, 0.0, 1.0}

\newcommand{\ISI}[1]{\mathrm{ISI}}

\definecolor{lime}{HTML}{A6CE39}
\DeclareRobustCommand{\orcidicon}{%
    \begin{tikzpicture}
    \draw[lime, fill=lime] (0,0) 
    circle [radius=0.16] 
    node[white] {{\fontfamily{qag}\selectfont \tiny ID}};
    \draw[white, fill=white] (-0.0625,0.095) 
    circle [radius=0.007];
    \end{tikzpicture}
    \hspace{-2mm}
}
\renewcommand{\orcid}[1]{\href{https://orcid.org/#1}{\textcolor[HTML]{A6CE39}{\orcidicon}}}

\newcommand{\uproman}[1]{\uppercase\expandafter{\romannumeral#1}}

\SetKwInput{KwInput}{Input}
\SetKwInput{KwOutput}{Output}

\SetCommentSty{mycommfont}
\theoremstyle{thmstyleone}%
\newtheorem{theorem}{Theorem}[section]%
\newtheorem{lemma}{Lemma}[section]%

\theoremstyle{thmstyletwo}%
\newtheorem{remark}{Remark}[section]%

\theoremstyle{thmstylethree}%
\newtheorem{definition}{Definition}[section]%

\hypersetup{bookmarksdepth=subsection}
\begin{document}

\title[Increased-order generalized synchronization in memristive neurons]{Increased-order generalized synchronization in memristive neurons: dissipativity analysis and data-driven control}

\author[1]{\fnm{Dennis} \sur{Mustafi\'{c}}\orcid{0009-0003-0772-612X}}\email{dennis.mustafic@fau.de}

\author[1,2]{\fnm{Marius E.} \sur{Yamakou}\orcid{0000-0002-2809-1739}}\email{marius.yamakou@fau.de}

\affil[1]{\orgdiv{Department of Data Science}, \orgname{Friedrich-Alexander-Universit\"at Erlangen-N\"urnberg}, \orgaddress{\street{N\"urnberger Str. 74}, \city{Erlangen}, \postcode{91052}, \country{Germany}}}

\affil[2]{\orgdiv{Department of Mathematics}, \orgname{Friedrich-Alexander-Universit\"at Erlangen-N\"urnberg}, \orgaddress{\street{Cauer Str. 11}, \city{Erlangen}, \postcode{91058}, \country{Germany}}}

\abstract{We investigate increased-order generalized synchronization between four- and five-dimensional memristive Hindmarsh--Rose neurons, focusing on the selection and stabilization of the synchronization map. An analytical construction combines a prescribed nonlinear map with adaptive control, retaining the memristive switching contribution explicitly in the error dynamics. Lyapunov analysis establishes asymptotic synchronization under uniform strict dissipativity and yields a switching-independent decrease estimate outside an explicit error neighborhood. Numerical simulations exhibit error decay beyond the sufficient convergence conditions. A complementary data-driven construction uses next-generation reservoir computing to learn predictors from uncoupled full-state trajectories and an affine synchronization map through dynamical consistency. The resulting sampled-data controller achieves small residual tracking errors without access to the governing equations. Relative to a reference map within the data-driven framework, the learned map reduces median control effort approximately eightfold while improving tracking accuracy. These results connect synchronization-map selection with conditional convergence guarantees and reduced control requirements.}


\maketitle

\section{Introduction}
\label{sec:introduction}

Synchronization is a fundamental phenomenon in nonlinear dynamics whereby interacting systems develop coordinated behavior. It establishes persistent relationships between phases, frequencies, or state trajectories, even when individual dynamics remain irregular. The central questions concern these relationships' existence, stability, and the coupling mechanisms that suppress deviations. These questions are particularly relevant when systems differ in parameters, timescales, or dimensions. Neural coordination has functional and pathological consequences: organized activity supports normal brain function \cite{Buzsaki.2006}, whereas excessive synchronization is associated with epileptic seizures \cite{Lehnertz.2009}. Understanding synchronization therefore requires identifying its dynamical conditions and how coupling or control modifies them.

Chaotic synchronization shows that unstable trajectories can maintain a stable relationship. Pecora and Carroll demonstrated this in a drive--response configuration, where an autonomous chaotic system unidirectionally forces a response \cite{Pecora.1990}. In complete synchronization, state trajectories asymptotically coincide, requiring compatible state spaces of equal dimension \cite{yamakou2016ratcheting}. Generalized synchronization instead describes a functional dependence of the response state on the drive state \cite{Rulkov.1995,Kocarev.1996}. Synchronization then corresponds to decaying deviations from a relation in the joint phase space. This viewpoint separates the complexity of synchronized motion from the relation's stability, providing a framework for nonidentical systems.

Unequal phase-space dimensions introduce a structural issue. Reduced-order synchronization concerns a response with fewer variables than the drive \cite{Femat.2002,Kobiolka.2025}, whereas increased-order generalized synchronization (IOGS) concerns a response with additional degrees of freedom \cite{Ogunjo.2013,Ojo.2015,Miao.2009}. IOGS must account for response variables without drive counterparts. Here, the relationship is designed: a target map is selected, and coupling drives the response toward its graph. The map determines the target motion and its mismatch with the intrinsic response dynamics. Its selection therefore affects the error equations, convergence conditions, and required forcing.

Memristive neuron models combine dimensional mismatch with memory, multiple timescales, and parameter uncertainty. The Hindmarsh--Rose (HR) model describes spiking and bursting by coupling membrane potential, fast recovery, and slow adaptation \cite{HindmarshRose.1984}. Memristive extensions add an internal state mediating activity-dependent feedback and modifying neuronal dynamics and synchronization \cite{Zhang.2024,Yamakou.2020}. Additional adaptation processes increase the dimension and introduce dynamical pathways \cite{Kobiolka.2025}. We consider a four-dimensional memristive HR drive and a five-dimensional response with an additional slow adaptation variable. Both include a switchable memristive element. This construction allows us to study coordination across unequal dimensions while retaining the memristive nonlinearity in the stability analysis.

Previous studies established increased-order synchronization using prescribed linear maps and active control with known parameters \cite{Ogunjo.2013,Ojo.2015}. Adaptive extensions address uncertain parameters, controlling additional response coordinates to reproduce a drive variable \cite{Al-sawalha.2011} or follow a prescribed function of the drive state \cite{Miao.2009}. For identical memristive HR neurons, Lyapunov analysis established local asymptotic stability of complete synchronization under dissipative memristive switching and practical stability with an explicit ultimate bound otherwise \cite{Babaeian.2026}. Reduced-order adaptive synchronization has also been studied for memristive HR neurons with parameter mismatch \cite{Kobiolka.2025}. These results motivate studying increased-order synchronization with a nonlinear target map while retaining switching explicitly in the error dynamics. The problem balances feedback-induced dissipation against retained memristive dynamics while compensating parameter mismatch through adaptation.

A complementary difficulty arises when governing equations are unavailable or poorly known. Reconstructing neuronal dynamics from recordings is challenging \cite{Durstewitz.2023}, and data-driven methods enable prediction and control from observed trajectories \cite{Brunton.2022}. Learned models serve as dynamical replicas and observers in synchronization research. A reservoir computer trained on a chaotic system can synchronize with it through an observed signal \cite{Weng.2019}. Reservoir observers infer unmeasured variables \cite{Lu.2017}, enabling physical-response synchronization with reduced coupling strength \cite{Nazerian.2023}. Reservoir observers of memristive HR neurons with different dimensions have also achieved reduced-order synchronization \cite{Kobiolka.2025}. These approaches demonstrate the usefulness of learned dynamics and raise a further question: how can they inform the target synchronization relationship?

Next generation reservoir computing (NG-RC) connects learned prediction with feedback control. NG-RC models compensate nonlinear dynamics and impose prescribed tracking-error decay in chaotic maps \cite{Kent.2024}, with related schemes implemented on embedded hardware for a physical chaotic circuit \cite{Kent.2024b}. Learning the synchronization map adds a geometric requirement: approximate compatibility with both systems' evolution. A related principle identifies conjugate mappings between Poincar\'{e} maps by enforcing consistency between coordinate transformation and dynamical evolution \cite{Bramburger.2021}. Here, we seek an approximately invariant graph under learned uncoupled dynamics. Advancing the drive and then mapping should approximately agree with mapping first and advancing the response. This criterion provides a self-supervised objective from uncoupled trajectories and measures the dynamical mismatch requiring control compensation. Exact conjugacy is not required.

We develop analytical and data-driven approaches centered on synchronization-map selection. The analytical approach prescribes a nonlinear map targeting all five response components while preserving the memristive coordinate. Feedback controllers and adaptive parameter-update laws retain switching explicitly in the error dynamics. Under the stated boundedness and existence assumptions, case-wise Lyapunov analysis proves asymptotic IOGS and bounded parameter-estimation errors when the memristive contribution is uniformly strictly dissipative. Without this condition, a switching-independent estimate gives an explicit error radius outside which the adaptive Lyapunov function decreases, without asserting ultimate boundedness. Numerical simulations additionally exhibit synchronization-error decay along a trajectory that violates the strict dissipativity hypothesis, distinguishing the observed behavior from the suffcient conditions established analytically.

The data-driven approach trains memoryless NG-RC predictors on uncoupled full-state recordings with fixed response parameters. An affine synchronization map is learned through dynamical consistency, with variance and covariance penalties discouraging collapsed targets. The predictors and map define sampled-data coupling with additive actuation of every response component; governing equations generate and assess the experiments but are unavailable to the learning and control algorithms. The response tracks the learned target with small residual error. Within this framework, the learned map reduces median control effort approximately eightfold relative to the reference correspondence while improving tracking accuracy. These results illustrate complementary roles of map selection: retaining tractable error dynamics for conditional convergence guarantees and promoting dynamical compatibility to reduce forcing. The constructions use different maps and assumptions; the data-driven results provide neither a stability proof nor a direct performance comparison with the analytical controller.

The paper is organized as follows. \Cref{sec:model} presents the neuronal models and uncontrolled dynamics. \Cref{sec:iogs} develops the nonlinear synchronization map, controllers, and adaptive laws, and \Cref{sec:main_results} establishes the analytical estimates. \Cref{sec:numerical_results} relates numerical behavior to the theoretical conditions. \Cref{sec:data-driven} presents the learned predictors, synchronization map, and sampled-data control. Finally, \Cref{sec:conclusion} summarizes the findings and open directions.

\section{Mathematical Models}
\label{sec:model}

To investigate increased-order generalized synchronization, we consider
two extensions of the Hindmarsh--Rose (HR) model, a paradigmatic
neuron model of established biological relevance: a four-dimensional
memristive HR drive system and a controlled five-dimensional memristive
HR response system with parameter mismatch. The underlying classical
three-dimensional HR model \cite{HindmarshRose.1984} describes the
interaction between the membrane potential, a fast recovery variable,
and a slow adaptation current. Both the drive and response incorporate
a switchable memristive autapse, while the response additionally
includes an adaptation variable associated with slower ionic processes.
This construction allows us to examine increased-order synchronization in the presence
of both unequal phase-space dimensions and uncertain response parameters.

\begin{figure*}
\centering
\includegraphics[width=0.65\textwidth]{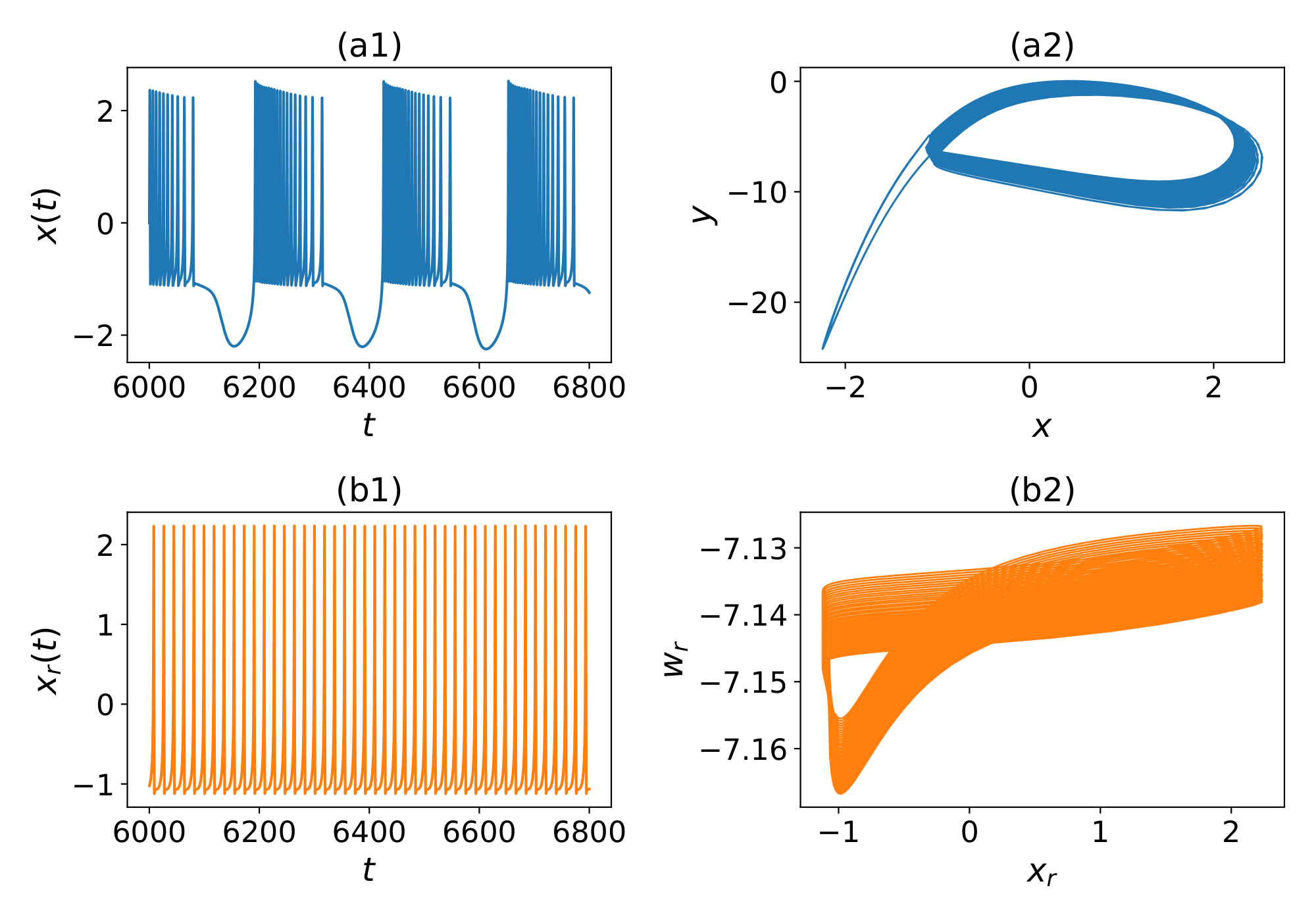}
\caption{Uncontrolled dynamics at $I=5.1$.
(a1) Drive membrane potential $x(t)$ and
(a2) trajectory projection onto the $(x,y)$ plane.
(b1) Response membrane potential $x_r(t)$ and
(b2) trajectory projection onto the $(x_r,w_r)$ plane.
Voltage traces are shown for $t\in[6000,6800]$.
All other parameter values are as specified in the main text.}
\label{fig:model_dynamics}
\end{figure*}

\subsection{Drive and response neuronal systems}
\label{subsec:drive_response}

Let $\mathbf{X}_d=(x,y,z,u)^\top\in\mathbb{R}^4$ denote the
drive state, where $x$ is the membrane potential, $y$ is the fast
recovery variable, $z$ represents the slow adaptation current, and
$u$ is the internal state of the switchable memristor \cite{Zhang.2024}. The drive
dynamics are given by
\begin{equation}
\begin{aligned}
&\dot{\mathbf{X}}_d=F_d(\mathbf{X}_d)=(F_1,F_2,F_3,F_4)^\top,\\[1mm]
&\begin{cases}
F_1:=y-ax^3+bx^2+k(h+fu^2)x-z+I,\\[1mm]
F_2:=c-px^2-y,\\[1mm]
F_3:=r[s(x-x_0)-z],\\[1mm]
F_4:=-\xi u+mG(u)+x,
\end{cases}
\end{aligned}
\label{eq:drive_system}
\end{equation}
where
\begin{equation}
G(u):=|u+1|-|u-1|
=
\begin{cases}
-2, & u<-1,\\
2u, & -1\leq u\leq 1,\\
2, & u>1.
\end{cases}
\label{eq:memristor_nonlinearity}
\end{equation}
The term $k(h+fu^2)x$ describes the memristive autaptic feedback
\cite{Zhang.2024}. Its strength is controlled by $k$,
while $h$ and $f$ determine the state-independent and
state-dependent contributions to the feedback coefficient,
respectively. The memristor state is driven by the membrane
potential and evolves under linear relaxation, $-\xi u$, and
the continuous piecewise-linear switching term $mG(u)$. Thus,
the membrane potential affects the internal memristor state,
which in turn modulates the voltage dynamics, introducing an
activity-dependent feedback mechanism. The parameter $m$ controls
the strength of the switching contribution.

The remaining terms retain their usual HR interpretation
\cite{HindmarshRose.1984}. The coefficients $a$ and $b$ determine
the cubic and quadratic voltage nonlinearities, $c$ and $p$
govern the fast recovery dynamics, and $I$ is the applied current.
The adaptation variable $z$ relaxes towards $s(x-x_0)$ at a rate
determined by $r$ and feeds back into the voltage equation
through $-z$.

\begin{figure*}[ht]
\centering
\includegraphics[width=0.63\textwidth]{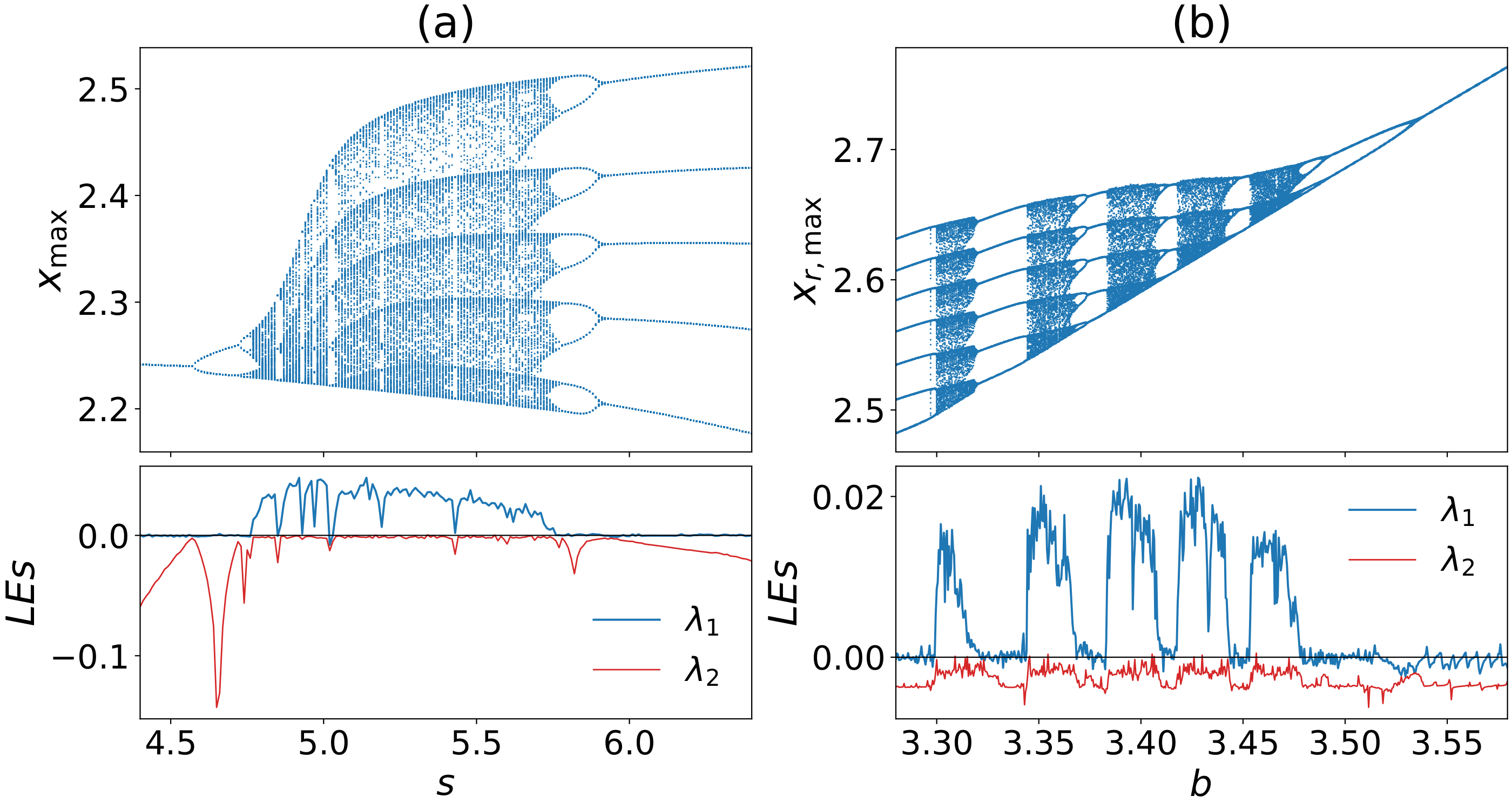}
\caption{Bifurcation diagrams (top) and the corresponding two
largest Lyapunov exponents $\lambda_1\geq\lambda_2$ (bottom)
for the uncontrolled drive and response systems, respectively.
(a) Post-transient local maxima of $x(t)$ versus
$s\in[4.40,6.40]$, with $I=5.0$ and $r=0.02$.
(b) Post-transient local maxima of $x_r(t)$ versus
$b\in[3.28,3.58]$. All other parameter values are as specified in the main text.}
\label{fig:bifurcation_lyapunov}
\end{figure*}

The response state is
$\mathbf{X}_r=(x_r,y_r,z_r,u_r,w_r)^\top\in\mathbb{R}^5$.
In addition to the four variables corresponding to those of the
drive, the response contains an additional slow adaptation variable
$w_r$, which phenomenologically represents the influence of slower
ionic processes \cite{Kobiolka.2025,Yamakou.2020}. Its controlled
dynamics are
\begin{equation}
\begin{cases}
\dot{x}_r
=
y_r-ax_r^3+bx_r^2
+\bar{k}(h+fu_r^2)x_r\\
\qquad{}-z_r+I+v_1,\\[1mm]
\dot{y}_r
=
c-px_r^2-y_r-\bar{\sigma}w_r+v_2,\\[1mm]
\dot{z}_r
=
r[\bar{s}(x_r-x_0)-z_r]+v_3,\\[1mm]
\dot{u}_r
=
-\bar{\xi}u_r+mG(u_r)+x_r+v_4,\\[1mm]
\dot{w}_r
=
\bar{\eta}[\gamma(y_r-y_0)-\rho w_r]+v_5.
\end{cases}
\label{eq:response_system}
\end{equation}
Here, $\mathbf{v}=(v_1,\ldots,v_5)^\top$ is the control vector
to be designed in the subsequent synchronization analysis. 
The additional variable $w_r$ phenomenologically represents calcium
exchange between intracellular stores and the cytoplasm \cite{falcke2000modeling,ngouonkadi2016bifurcations}.
It is driven by the deviation of the recovery variable $y_r$ from
the reference level $y_0$ and feeds back into the recovery equation
through the term $-\bar{\sigma}w_r$. The parameters $\gamma$ and $\rho$ weight
the recovery-dependent forcing and linear relaxation terms,
respectively, while $\bar{\eta}$ scales the intrinsic dynamics
of $w_r$. For $\bar{\eta}(t)>0$ and $\rho>0$, the linear decay
term $-\bar{\eta}(t)\rho w_r$ defines an instantaneous relaxation
rate $\bar{\eta}(t)\rho$, with corresponding timescale
$[\bar{\eta}(t)\rho]^{-1}$. Consequently, the response extends
the drive through an additional dynamical feedback pathway.

\subsection{Parameter mismatch and nominal values}
\label{subsec:parameters}

The time-dependent quantities
$\bar{k}$, $\bar{\sigma}$, $\bar{s}$, $\bar{\xi}$, and
$\bar{\eta}$ are adjustable response coefficients updated by
the adaptive laws. They serve as estimates of the fixed nominal
target parameters $k$, $\sigma$, $s$, $\xi$, and $\eta$,
respectively. The parameters $k$, $s$, and $\xi$ occur in the
drive and have corresponding adjustable coefficients in the
response. The parameters $\sigma$ and $\eta$ characterize the
additional response pathway and have no counterparts in the
four-dimensional drive.

Using the convention of estimate minus target, we define
\begin{equation}
\begin{aligned}
\hat{k}&=\bar{k}-k,
&
\hat{\sigma}&=\bar{\sigma}-\sigma,
&
\hat{s}&=\bar{s}-s,
\\
\hat{\xi}&=\bar{\xi}-\xi,
&
\hat{\eta}&=\bar{\eta}-\eta.
\end{aligned}
\label{eq:param_errors}
\end{equation}
The nominal target values specify the numerical model and
provide reference values for evaluating the parameter-estimation
errors; they are not supplied to the adaptive controller.
Unless stated otherwise, the known parameters are fixed at
$a=1.0$, $b=3.0$, $c=1.0$, $p=5.0$,
$f=0.2$, $h=0.3$, $m=0.5$, $r=0.006$,
$x_0=-1.56$, $I=3.1$, $y_0=-1.619$,
$\gamma=3.0$, and $\rho=0.9573$.
The parameters $y_0$, $\gamma$, and $\rho$ enter only the
response system; the other known parameters are shared by
both systems. The nominal target values are
$k=0.87$, $\sigma=0.0278$, $s=5.2$,
$\xi=1.0$, and $\eta=0.0009$.
The drive evolves with fixed values of $k$, $s$, and $\xi$,
whereas the five adjustable response coefficients evolve
according to the adaptive laws.

\subsection{Models dynamics}
\label{subsec:model_dynamics}

To characterize the uncontrolled dynamics, we integrate
Eqs.~\eqref{eq:drive_system} and \eqref{eq:response_system}
independently with $\mathbf{v}=\mathbf{0}$ and
$(\bar{k},\bar{\sigma},\bar{s},\bar{\xi},\bar{\eta})
=(k,\sigma,s,\xi,\eta)$, using the initial conditions
\[
\begin{aligned}
\mathbf{X}_d(0)&=(0.1,0.1,0.1,0.1)^\top,
\\
\mathbf{X}_r(0)&=(0.5,0.5,0.5,0.5,0)^\top.
\end{aligned}
\]
Integration uses the adaptive Dormand--Prince Runge--Kutta
method of order $5(4)$ \cite{Dormand.1980}, implemented as
\texttt{RK45} in \texttt{SciPy} \cite{Virtanen.2020},
with relative and absolute tolerances $10^{-6}$ and
$10^{-8}$, respectively.

Figure~\ref{fig:model_dynamics} illustrates some dynamical behaviors. After discarding the initial $6000$ time units,
the voltage traces over $t\in[6000,6800]$ exhibit bursting
in the drive and sustained spiking in the response.
The corresponding trajectory projections illustrate
the geometry of the dynamics.

To examine some parameter-dependent transitions, we construct
numerical bifurcation diagrams from post-transient local
maxima of $x(t)$ and $x_r(t)$ and estimate the two largest
Lyapunov exponents, $\lambda_1\geq\lambda_2$, from the
associated variational equations using the Benettin algorithm.
Figure~\ref{fig:bifurcation_lyapunov} shows the dependence
of the drive dynamics on $s$ and the response dynamics on $b$.
Intervals with $\lambda_1>0$ provide numerical evidence of
chaos and are interspersed with windows consistent with
periodic dynamics. 


\section{The IOGS problem: an analytical approach}
\label{sec:iogs}

We study an IOGS problem between the four-dimensional memristive HR drive system and the five-dimensional response system with parameter uncertainty. We first define a nonlinear map from the drive state space to the response state space, specifying a synchronization manifold and its associated error variables. This construction assigns a target to each response variable, including the additional adaptation variable $w_r$, which has no counterpart in the drive. We then formulate feedback controllers and adaptive parameter-update laws and derive the closed-loop error dynamics, retaining the memristive switching contribution explicitly. A case-wise analysis of this contribution identifies the switching configurations that ensure strict dissipation in the memristive error channel. These estimates form the basis for establishing our two main theoretical results: asymptotic IOGS under a uniform dissipativity condition on the memristive contribution, and a practical Lyapunov-type estimate without this condition. The latter provides an explicit error threshold above which the adaptive Lyapunov function is strictly decreasing, but does not by itself establish ultimate boundedness of the synchronization error.

\subsection{Nonlinear increased-order synchronization manifold}
\label{subsec:manifold}
We define IOGS through a nonlinear map $\varphi:\mathbb R^4\to\mathbb R^5$, which in the present work we choose as
\begin{equation}
\begin{aligned}
&\varphi(x,y,z,u)=(\varphi_1,\varphi_2,\varphi_3,\varphi_4,\varphi_5)^\top,\\
&\begin{cases}
\varphi_1:=x+\alpha_1\sin x,\\
\varphi_2:=y+\alpha_2x^2+\alpha_3u^2,\\
\varphi_3:=z+\alpha_4xy,\\
\varphi_4:=u,\\
\varphi_5:=\beta_0+\beta_1x+\beta_2y+\beta_3z\\
\qquad{}+\beta_4xu+\beta_5\cos u.
\end{cases}
\end{aligned}
\label{eq:nonlinear_phi}
\end{equation}

Here $\alpha_1,\alpha_2,\alpha_3,\alpha_4,\beta_0,\beta_1,\beta_2,\beta_3,\beta_4,\beta_5\in\mathbb{R}$. At least one nonlinear coefficient is assumed to be nonzero: $(\alpha_1,\alpha_2,\alpha_3,\alpha_4,\beta_4,\beta_5)\neq 0$.

The fourth component is chosen as $\varphi_4(x,y,z,u)=u$. This choice is crucial because it preserves the natural memristive error $e_4=u_r-u$. Thus the switchable memristor contribution remains explicit in the error dynamics and is not hidden by the controller.

The IOGS error is defined by
\begin{equation}
\mathbf{e}=\mathbf{X}_r-\varphi(\mathbf{X}_d),
\label{eq:error_definition}
\end{equation}
that is,
\begin{equation}
\begin{cases}
\begin{aligned}
e_1
&=
x_r-\big(x+\alpha_1\sin x\big),\\[1mm]
e_2
&=
y_r-\big(y+\alpha_2x^2+\alpha_3u^2\big),\\[1mm]
e_3
&=
z_r-\big(z+\alpha_4xy\big),\\[1mm]
e_4
&=
u_r-u,\\[1mm]
e_5
&=
w_r-\big(\beta_0+\beta_1x+\beta_2y+\beta_3z\\
&\quad+\beta_4xu+\beta_5\cos u\big).
\end{aligned}
\end{cases}
\label{eq:error_components}
\end{equation}
The associated IOGS manifold is
\begin{equation}
\mathcal M_{\rm IOGS}
=
\left\{
(\mathbf{X}_d,\mathbf{X}_r)\in\mathbb R^4\times\mathbb R^5:
\mathbf{X}_r=\varphi(\mathbf{X}_d)
\right\}.
\label{eq:iogs_manifold}
\end{equation}

\begin{definition}\label{def_1}
\textit{The response system Eq. \eqref{eq:response_system} is said to achieve \textit{increased-order generalized synchronization} with the drive system Eq. \eqref{eq:drive_system} if}
\begin{equation}
\lim_{t\to\infty}
\left\|
\mathbf{X}_r(t)-\varphi(\mathbf{X}_d(t))
\right\|
=0.
\label{eq:iogs_definition}
\end{equation}
\end{definition}
\begin{definition}\label{def_2}
\textit{The response system Eq. \eqref{eq:response_system} is said to achieve \textit{practical increased-order generalized synchronization} with the drive system Eq. \eqref{eq:drive_system} if there exists \(\delta>0\) such that}
\begin{equation}
\limsup_{t\to\infty}
\left\|
\mathbf{X}_r(t)-\varphi(\mathbf{X}_d(t))
\right\|
\le \delta.
\label{eq:practical_iogs_definition}
\end{equation}
\end{definition}

\subsection{Controller design and adaptive laws}
\label{subsec:controller_design}

Since $e_i=X_{r,i}-\varphi_i(\mathbf{X}_d)$, we have
\begin{equation}
\dot e_i=\dot X_{r,i}-\dot\varphi_i(\mathbf{X}_d),
\qquad i=1,\ldots,5.
\label{eq:error_dynamics}
\end{equation}

The derivatives of the components of the nonlinear synchronization map are
\begin{equation}
\begin{cases}
\begin{aligned}
\dot\varphi_1
&=
(1+\alpha_1\cos x)F_1,\\[1mm]
\dot\varphi_2
&=
F_2+2\alpha_2xF_1+2\alpha_3uF_4,\\[1mm]
\dot\varphi_3
&=
F_3+\alpha_4yF_1+\alpha_4xF_2,\\[1mm]
\dot\varphi_4
&=
F_4,\\[1mm]
\dot\varphi_5
&=
(\beta_1+\beta_4u)F_1
+\beta_2F_2
+\beta_3F_3\\
&\quad+(\beta_4x-\beta_5\sin u)F_4.
\end{aligned}
\end{cases}
\label{eq:phi_derivatives}
\end{equation}

We introduce the response regressors
\begin{equation}
\begin{cases}
\Phi_1=(h+fu_r^2)x_r,\\
\Phi_2=w_r,\\
\Phi_3=r(x_r-x_0),\\
\Phi_4=u_r,\\
\Phi_5=\gamma(y_r-y_0)-\rho w_r.
\end{cases}
\label{eq:regressors}
\end{equation}

Let $\lambda_1,\lambda_2,\lambda_3,\lambda_5>0$ and choose control inputs as
\begin{equation}
\begin{cases}
\begin{aligned}
v_1={}&
\dot\varphi_1
-
\Big[
y_r-a x_r^3+b x_r^2+k\Phi_1-z_r+I
\Big]\\
&-\lambda_1e_1,
\\[1mm]
v_2={}&
\dot\varphi_2
-
\Big[
c-px_r^2-y_r-\sigma\Phi_2
\Big]
-\lambda_2e_2,
\\[1mm]
v_3={}&
\dot\varphi_3
-
r\big[s(x_r-x_0)-z_r\big]
-\lambda_3e_3,
\\[1mm]
v_4={}&
x-x_r,
\\[1mm]
v_5={}&
\dot\varphi_5
-
\eta\Phi_5
-\lambda_5e_5.
\end{aligned}
\end{cases}
\label{eq:controllers_explicit}
\end{equation}
The fourth controller \(v_4=x-x_r\) cancels only the direct mismatch induced by the membrane-potential term in the \(u\)-equation. It does not cancel the memristive switching difference \(G(u_r)-G(u)\).

The adaptive update laws are selected as
\begin{equation}
\begin{cases}
\begin{aligned}
\dot{\bar{k}}
&=
-\Phi_1e_1,\\
\dot{\bar{\sigma}}
&=
\Phi_2e_2,\\
\dot{\bar{s}}
&=
-\Phi_3e_3,\\
\dot{\bar{\xi}}
&=
\Phi_4e_4,\\
\dot{\bar{\eta}}
&=
-\Phi_5e_5.
\end{aligned}
\end{cases}
\label{eq:update_laws}
\end{equation}
Since the target parameters are constant, the parameter-error dynamics satisfy
\begin{equation}
\begin{cases}
\begin{aligned}
\dot{\hat{k}}=\dot{\bar{k}},\\
\dot{\hat{\sigma}}=\dot{\bar{\sigma}},\\
\dot{\hat{s}}=\dot{\bar{s}},\\
\dot{\hat{\xi}}=\dot{\bar{\xi}},\\
\dot{\hat{\eta}}=\dot{\bar{\eta}}.
\end{aligned}
\end{cases}
\label{eq:update_laws_2}
\end{equation}

Substituting Eq. \eqref{eq:controllers_explicit} into the error equations in Eq. \eqref{eq:error_dynamics} gives the closed-loop adaptive error system
\begin{equation}
\begin{cases}
\dot e_1
=
-\lambda_1e_1+\hat{k}\Phi_1,\\[1mm]
\dot e_2
=
-\lambda_2e_2-\hat{\sigma}\Phi_2,\\[1mm]
\dot e_3
=
-\lambda_3e_3+\hat{s}\Phi_3,\\[1mm]
\dot e_4
=
-e_4+m\big[G(u_r)-G(u)\big]-\hat{\xi}\Phi_4,\\[1mm]
\dot e_5
=
-\lambda_5e_5+\hat{\eta}\Phi_5.
\end{cases}
\label{eq:closed_loop_error_system_adaptive}
\end{equation}

\subsection{Case-wise analysis of the memristive contribution}
\label{subsec:casewise}

The memristive contribution to the Lyapunov derivative is
\begin{equation}
\begin{aligned}
&\Psi(e_4;u,u_r,m):=
\\
&\quad\left[-e_4+m\big(G(u_r)-G(u)\big)\right]e_4,
\quad e_4=u_r-u.
\end{aligned}
\label{eq:Psi_definition}
\end{equation}
Since $G(u)$ is given piecewise by Eq. \eqref{eq:memristor_nonlinearity}, $\Psi$ admits the following case-wise representation:
\begingroup\predisplaypenalty=0
\begin{equation}
\begin{aligned}
&\Psi(e_4;u,u_r,m)
=\\
&\left\{\begin{array}{@{}l@{\;\;}l@{}}
-e_4^2,
& \mathcal C_1:\ u\ge1,\\
& \phantom{\mathcal C_1:\ }u_r\ge1,\\[1mm]
(2m-1)e_4^2+2m(u-1)e_4,
& \mathcal C_2:\ u\ge1,\\
& \phantom{\mathcal C_2:\ }{-1}<u_r<1,\\[1mm]
-e_4^2-4me_4,
& \mathcal C_3:\ u\ge1,\\
& \phantom{\mathcal C_3:\ }u_r\le-1,\\[1mm]
-e_4^2+2m(1-u)e_4,
& \mathcal C_4:\ -1<u<1,\\
& \phantom{\mathcal C_4:\ }u_r\ge1,\\[1mm]
(2m-1)e_4^2,
& \mathcal C_5:\ -1<u<1,\\
& \phantom{\mathcal C_5:\ }{-1}<u_r<1,\\[1mm]
-e_4^2-2m(1+u)e_4,
& \mathcal C_6:\ -1<u<1,\\
& \phantom{\mathcal C_6:\ }u_r\le-1,\\[1mm]
-e_4^2+4me_4,
& \mathcal C_7:\ u\le-1,\\
& \phantom{\mathcal C_7:\ }u_r\ge1,\\[1mm]
(2m-1)e_4^2+2m(u+1)e_4,
& \mathcal C_8:\ u\le-1,\\
& \phantom{\mathcal C_8:\ }{-1}<u_r<1,\\[1mm]
-e_4^2,
& \mathcal C_9:\ u\le-1,\\
& \phantom{\mathcal C_9:\ }u_r\le-1.
\end{array}\right.
\end{aligned}
\label{eq:Psi_piecewise}
\end{equation}
\endgroup

For $e_4\neq0$, define the normalized memristive coefficient
\begin{equation}
\Theta(u,u_r,m)
=
\frac{\Psi(e_4;u,u_r,m)}{e_4^2}.
\label{eq:normalized_coeff}
\end{equation}
\begin{figure}[b]
\centering
\includegraphics[width=\columnwidth]{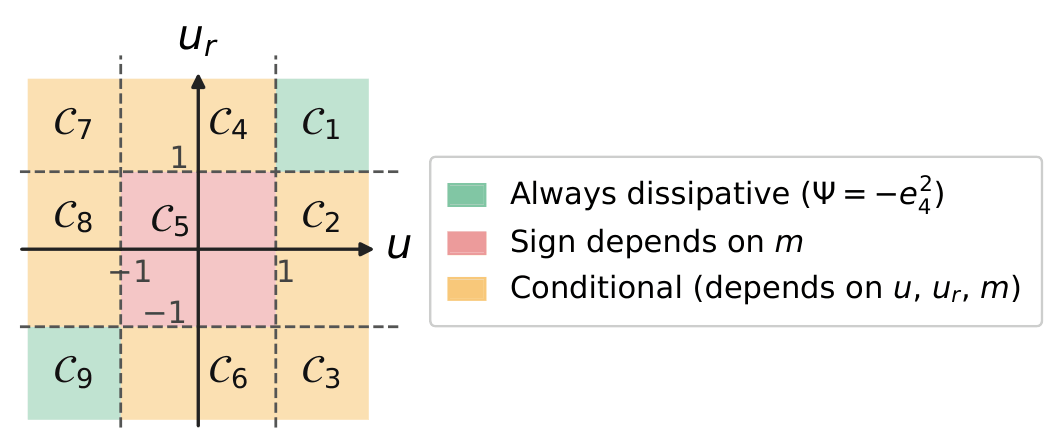}
\caption{Partition of the $(u,u_r)$ plane into the nine switching cases $\mathcal{C}_1,\ldots,\mathcal{C}_9$ of $\Psi$ defined in Eq. \eqref{eq:Psi_piecewise}. Green regions ($\mathcal{C}_1$, $\mathcal{C}_9$) are always dissipative, with $\Psi=-e_4^2$. The red region ($\mathcal{C}_5$) has a sign that depends on $m$. Amber regions are conditional, with dissipativity depending on the trajectory $(u(t),u_r(t))$ and $m$.}
\label{fig:switching_regimes}
\end{figure}
The trajectory is said to remain in a strictly dissipative memristive switching regime if there exists $\mu>0$ such that
\begin{equation}
\Theta(u(t),u_r(t),m)\le -\mu,
\qquad t\ge0.
\label{eq:dissipative_set_condition}
\end{equation}
Equivalently, $\Psi(e_4(t);u(t),u_r(t),m)\le -\mu e_4^2(t)$ for all $t\ge0$.

More explicitly, the dissipative switching set associated with a given $\mu>0$ is
\begin{equation}
\mathcal D_\mu
=
\left\{
(u,u_r):
\Psi(e_4;u,u_r,m)\le -\mu e_4^2
\right\}.
\label{eq:dissipative_set}
\end{equation}
Using Eq. \eqref{eq:Psi_piecewise}, this set is the union of the following case-wise regions:
\begin{equation}
\begin{aligned}
&\mathcal D_\mu
=
\mathcal C_1
\cup
\mathcal C_9
\\
&\quad\cup
\big\{
\mathcal C_2:
(2m-1)e_4^2+2m(u-1)e_4
\le -\mu e_4^2
\big\}
\\
&\quad\cup
\big\{
\mathcal C_3:
-e_4^2-4me_4
\le -\mu e_4^2
\big\}
\\
&\quad\cup
\big\{
\mathcal C_4:
-e_4^2+2m(1-u)e_4
\le -\mu e_4^2
\big\}
\\
&\quad\cup
\big\{
\mathcal C_5:
(2m-1)e_4^2
\le -\mu e_4^2
\big\}
\\
&\quad\cup
\big\{
\mathcal C_6:
-e_4^2-2m(1+u)e_4
\le -\mu e_4^2
\big\}
\\
&\quad\cup
\big\{
\mathcal C_7:
-e_4^2+4me_4
\le -\mu e_4^2
\big\}
\\
&\quad\cup
\big\{
\mathcal C_8:
(2m-1)e_4^2+2m(u+1)e_4
\le -\mu e_4^2
\big\}.
\end{aligned}
\label{eq:dissipative_cases}
\end{equation}
Thus, the purely outer cases $\mathcal C_1$ and $\mathcal C_9$ are always dissipative with $\mu\le1$. The remaining cases may be dissipative or non-dissipative depending on $m$ and on the instantaneous position of the trajectory within the corresponding switching region. Figure \ref{fig:switching_regimes} illustrates the nine switching cases in the $(u,u_r)$ plane, color-coded by their dissipativity character.

The complementary set $\mathcal N_\mu = \mathbb{R}^2\setminus\mathcal D_\mu$ is referred to as the non-dissipative memristive switching set. In this set, the strict estimate $\Psi(e_4;u,u_r,m)\le -\mu e_4^2$ is not guaranteed. 
The following section establishes our main theoretical results,
first under the uniform dissipativity condition and then without
this restriction.

\section{Main theoretical results on IOGS}
\label{sec:main_results}

We now establish the two main theoretical results of this paper using the
closed-loop error dynamics and the estimates for the memristive
contribution derived above.

\subsection{Adaptive asymptotic IOGS under dissipative memristive switching}
\label{subsec:adaptive_asymptotic}
When the trajectory $(u(t),u_r(t))$ remains in $\mathcal D_\mu$
for some fixed $\mu>0$, the adaptive laws and the uniform dissipativity estimate yield a Lyapunov inequality that ensures asymptotic convergence of the synchronization error. The following theorem states the required assumptions and establishes boundedness of the parameter-estimation errors.

\begin{theorem}
\label{thm:adaptive_asymptotic}
Consider the drive system Eq. \eqref{eq:drive_system}, the controlled response system Eq. \eqref{eq:response_system}, and the nonlinear synchronization map Eq. \eqref{eq:nonlinear_phi}. Let the controllers be given by Eq. \eqref{eq:controllers_explicit}, and let the adaptive laws be given by Eq. \eqref{eq:update_laws}. Assume that the drive trajectory is bounded and that the closed-loop solution exists for all $t\ge0$. Assume further that the memristive switching trajectory remains in the dissipative set $\mathcal D_\mu$ for some constant $\mu>0$, that is, $(u(t),u_r(t))\in\mathcal D_\mu$ for all $t\ge0$, equivalently
\begin{equation}
\Psi(e_4(t);u(t),u_r(t),m)
\le
-\mu e_4^2(t),
\qquad t\ge0.
\label{eq:dissipative_assumption_adaptive}
\end{equation}
Then the increase-order synchronization error satisfies
\begin{equation}
\lim_{t\to\infty}\|\mathbf{e}(t)\|=0.
\label{eq:adaptive_asymptotic_conclusion}
\end{equation}
Consequently, the response system achieves increased-order generalized synchronization with the drive system:
\begin{equation}
\lim_{t\to\infty}
\left\|
\mathbf{X}_r(t)-\varphi(\mathbf{X}_d(t))
\right\|
=0.
\label{eq:adaptive_iogs_conclusion}
\end{equation}
Moreover, the parameter-estimation errors remain bounded.
\end{theorem}

\begin{proof}
Consider the adaptive Lyapunov function
\begin{equation}
\begin{split}
V
={}&
\frac12
\Big(
e_1^2+e_2^2+e_3^2+e_4^2+e_5^2\\
&+\hat{k}^2+\hat{\sigma}^2+\hat{s}^2
+\hat{\xi}^2+\hat{\eta}^2
\Big)\ge0.
\end{split}
\label{eq:adaptive_lyapunov}
\end{equation}
Along trajectories of Eq. \eqref{eq:closed_loop_error_system_adaptive},
\begin{align}
\dot V
={}&
e_1\dot e_1
+e_2\dot e_2
+e_3\dot e_3
+e_4\dot e_4
+e_5\dot e_5
\nonumber\\
&+
\hat{k}\dot{\bar{k}}
+\hat{\sigma}\dot{\bar{\sigma}}
+\hat{s}\dot{\bar{s}}
+\hat{\xi}\dot{\bar{\xi}}
+\hat{\eta}\dot{\bar{\eta}}.
\label{eq:Vdot_expanded_start}
\end{align}
Substituting Eq. \eqref{eq:closed_loop_error_system_adaptive} gives
\begin{align}
\dot V
={}&
-\lambda_1e_1^2
+\Phi_1\hat{k}e_1
-\lambda_2e_2^2
-\Phi_2\hat{\sigma}e_2
\nonumber\\
&-\lambda_3e_3^2
+\Phi_3\hat{s}e_3
-e_4^2
\nonumber\\
&+\left[m\big(G(u_r)-G(u)\big)-\Phi_4\hat{\xi}\right]e_4
\nonumber\\
&-\lambda_5e_5^2
+\Phi_5\hat{\eta}e_5
\nonumber\\
&+
\hat{k}\dot{\bar{k}}
+\hat{\sigma}\dot{\bar{\sigma}}
+\hat{s}\dot{\bar{s}}
+\hat{\xi}\dot{\bar{\xi}}
+\hat{\eta}\dot{\bar{\eta}}.
\label{eq:Vdot_expanded}
\end{align}
Using the adaptive laws Eq. \eqref{eq:update_laws}, all parameter-error cross terms cancel:
\begin{equation}
\begin{aligned}
&\Phi_1\hat{k}e_1+\hat{k}\dot{\bar{k}}=0,
&&-\Phi_2\hat{\sigma}e_2+\hat{\sigma}\dot{\bar{\sigma}}=0,\\
&\Phi_3\hat{s}e_3+\hat{s}\dot{\bar{s}}=0,
&&-\Phi_4\hat{\xi}e_4+\hat{\xi}\dot{\bar{\xi}}=0,\\
&\Phi_5\hat{\eta}e_5+\hat{\eta}\dot{\bar{\eta}}=0.
\end{aligned}
\label{eq:cross_term_cancellation}
\end{equation}
Therefore,
\begin{equation}
\dot V
=
-\lambda_1e_1^2
-\lambda_2e_2^2
-\lambda_3e_3^2
-\lambda_5e_5^2
+
\Psi(e_4;u,u_r,m).
\label{eq:Vdot_after_cancellation}
\end{equation}
Since $(u(t),u_r(t))\in\mathcal D_\mu$, the case-wise definition Eq. \eqref{eq:dissipative_set} gives $\Psi(e_4;u,u_r,m)\le -\mu e_4^2$. Thus
\begin{equation}
\dot V
\le
-\lambda_1e_1^2
-\lambda_2e_2^2
-\lambda_3e_3^2
-\mu e_4^2
-\lambda_5e_5^2.
\label{eq:Vdot_asymptotic_adaptive}
\end{equation}
Let $\alpha=\min\{\lambda_1,\lambda_2,\lambda_3,\mu,\lambda_5\}>0$. Then
\begin{equation}
\dot V
\le
-\alpha\|\mathbf{e}\|^2.
\label{eq:Vdot_alpha_error}
\end{equation}
Hence $V(t)\le V(0)$ for all $t\ge0$. Therefore, the synchronization errors and the parameter-estimation errors are bounded. Since the drive trajectory is bounded (consistent with the existence of chaotic attractors in the HR neuron model, which is observed numerically) and $\mathbf{X}_r=\mathbf{e}+\varphi(\mathbf{X}_d)$, the response trajectory is bounded as well. Consequently, the regressors $\Phi_i$, $i=1,\ldots,5$ in Eq. \eqref{eq:regressors} are bounded. From Eq. \eqref{eq:closed_loop_error_system_adaptive}, it follows that $\dot{\mathbf{e}}\in L^\infty$.

\begin{lemma}[Barbalat's Lemma]
\label{lem:barbalat}
Let $g:[0,\infty)\to\mathbb{R}$ be uniformly continuous and suppose that $\int_0^t g(s)\,ds$ has a finite limit as $t\to\infty$. Then $\displaystyle \lim_{t \to \infty} g(t) = 0$.
\end{lemma}

Integrating Eq. \eqref{eq:Vdot_alpha_error} from $0$ to $t$, we obtain
\begin{equation*}
V(t)-V(0)
\le
-\alpha\int_0^t\|\mathbf{e}(\tau)\|^2\,d\tau.
\end{equation*}
Since $V(t)\ge0$, this implies $\int_0^t\|\mathbf{e}(\tau)\|^2\,d\tau\le\frac{V(0)}{\alpha}$, $t\ge0$. Hence, $\|\mathbf{e}(\cdot)\|^2\in L^1([0,\infty))$.

Let $g(t)=\|\mathbf{e}(t)\|^2$ so that
\begin{equation*}
\dot g=\frac{d}{dt}\|\mathbf{e}(t)\|^2 = 2 \mathbf{e}(t)^\top \dot{\mathbf{e}}(t).
\end{equation*}
Since $\mathbf{e}\in L^\infty$ and $\dot{\mathbf{e}}\in L^\infty$, the function $g$ is Lipschitz and uniformly continuous on $[0,\infty)$. By Barbalat's lemma, $\lim_{t\to\infty}\|\mathbf{e}(t)\|^2=0$, and by the continuity of the square-root function, we have $\lim_{t\to\infty}\|\mathbf{e}(t)\|=0$. Because $\mathbf{e}(t)=\mathbf{X}_r(t)-\varphi(\mathbf{X}_d(t))$, we conclude that
\begin{equation*}
\lim_{t\to\infty}\|\mathbf{X}_r(t)-\varphi(\mathbf{X}_d(t))\|=0.
\end{equation*}
This is precisely the condition in Definition \ref{def_1}. Hence the response system achieves increased-order generalized synchronization with the drive system. Boundedness of the parameter-estimation errors follows from the boundedness of $V(t)$.
\end{proof}

\subsection{Adaptive practical Lyapunov-type estimate under non-dissipative memristive switching}
\label{subsec:adaptive_practical}
When the trajectory $(u(t),u_r(t))$ lies outside $\mathcal D_\mu$,
the prescribed uniform dissipativity condition is not satisfied.
Nevertheless, the following theorem establishes a practical
Lyapunov-type estimate valid irrespective of the active switching
region and identifies an explicit error threshold above which
the adaptive Lyapunov function is strictly decreasing.

\begin{theorem}
\label{thm:adaptive_practical_lyapunov}
Consider the same drive-response systems, synchronization map, controllers, and adaptive laws as in Theorem \ref{thm:adaptive_asymptotic}. Suppose that the closed-loop solution exists for all $t\ge0$ and let $\varepsilon\in(0,2)$. Define
\begin{equation}
C=\frac{8m^2}{\varepsilon},
\qquad
\alpha
=
\min
\left\{
\lambda_1,\lambda_2,\lambda_3,\lambda_5,
1-\frac{\varepsilon}{2}
\right\}.
\label{eq:C_alpha_practical_adaptive}
\end{equation}
If $\alpha>0$, then the adaptive Lyapunov function Eq. \eqref{eq:adaptive_lyapunov} satisfies
\begin{equation}
\dot V(t)\le -\alpha\|\mathbf{e}(t)\|^2+C.
\label{eq:Vdot_practical_adaptive}
\end{equation}
Consequently, $\dot V(t)<0$ whenever $\|\mathbf{e}(t)\|>\sqrt{C/\alpha}$, so that the adaptive Lyapunov function decreases whenever the synchronization error lies outside the ball of radius
\begin{equation}
\delta
=
\sqrt{\frac{C}{\alpha}}
=
\sqrt{\frac{8m^2}{\varepsilon\alpha}}.
\label{eq:practical_radius_adaptive}
\end{equation}
Therefore, the closed-loop system satisfies a practical Lyapunov-type bound with respect to the increased-order synchronization error. Equivalently, the adaptive Lyapunov function is strictly decreasing outside the error neighborhood
\begin{equation}
\begin{split}
\mathcal B_\delta(\mathcal M_{\rm IOGS})
=
\left\{
(\mathbf{X}_d,\mathbf{X}_r):
\left\|
\mathbf{X}_r-\varphi(\mathbf{X}_d)
\right\|
\le \delta
\right\}.
\end{split}
\label{eq:practical_neighborhood}
\end{equation}
\end{theorem}

\begin{proof}
Using the adaptive Lyapunov function Eq. \eqref{eq:adaptive_lyapunov} and repeating the cancellation argument from the proof of Theorem \ref{thm:adaptive_asymptotic}, we again arrive at Eq. \eqref{eq:Vdot_after_cancellation}, that is,
\begin{equation*}
\dot V
=
-\lambda_1e_1^2
-\lambda_2e_2^2
-\lambda_3e_3^2
-\lambda_5e_5^2
+
\Psi(e_4;u,u_r,m).
\end{equation*}
It remains to bound the memristive contribution $\Psi(e_4;u,u_r,m)$ under non-dissipative switching. By Eq. \eqref{eq:Psi_definition},
\begin{equation*}
\Psi(e_4;u,u_r,m)
=
-e_4^2
+
m\big[G(u_r)-G(u)\big]e_4.
\end{equation*}
Since $|G(u)|\le2$ for all $u\in\mathbb R$, the triangle inequality gives $|G(u_r)-G(u)|\le|G(u_r)|+|G(u)|\le4$, so that
\begin{equation*}
m\big[G(u_r)-G(u)\big]e_4
\le
|m|\,|G(u_r)-G(u)|\,|e_4|
\le
4|m||e_4|.
\end{equation*}
By Young's inequality, for every $\varepsilon>0$,
\begin{equation*}
4|m||e_4|
\le
\frac{\varepsilon}{2}e_4^2
+
\frac{8m^2}{\varepsilon}.
\end{equation*}
Therefore,
\begin{equation*}
\Psi(e_4;u,u_r,m)
\le
-\left(1-\frac{\varepsilon}{2}\right)e_4^2
+
\frac{8m^2}{\varepsilon}.
\end{equation*}
Substituting this estimate into the expression for $\dot V$, we obtain
\begin{equation*}
\dot V
\le
-\lambda_1e_1^2
-\lambda_2e_2^2
-\lambda_3e_3^2
-\lambda_5e_5^2
-\left(1-\frac{\varepsilon}{2}\right)e_4^2
+
\frac{8m^2}{\varepsilon}.
\end{equation*}
Using the definitions of $C$ and $\alpha$ in Eq. \eqref{eq:C_alpha_practical_adaptive}, this becomes
\begin{equation*}
\dot V
\le
-\alpha\|\mathbf{e}\|^2+C,
\end{equation*}
which is precisely Eq. \eqref{eq:Vdot_practical_adaptive}. Consequently, if $\|\mathbf{e}(t)\|^2>C/\alpha$, then $\dot V(t)<0$. Since $\mathbf{e}(t)=\mathbf{X}_r(t)-\varphi(\mathbf{X}_d(t))$, this ball corresponds to the error neighborhood Eq. \eqref{eq:practical_neighborhood} with radius $\delta=\sqrt{C/\alpha}=\sqrt{8m^2/(\varepsilon\alpha)}$. Thus the closed-loop system satisfies a practical Lyapunov-type bound with respect to the increased-order synchronization error, which proves the stated estimate.
\end{proof}

\begin{remark}
The estimate in Theorem \ref{thm:adaptive_practical_lyapunov} identifies the neighborhood appearing in Definition \ref{def_2} as the natural practical synchronization neighborhood. However, without an additional ultimate-boundedness argument (which is beyond the scope of the current manuscript), it should be interpreted as a Lyapunov-type decrease condition rather than as a proof of the full $\limsup$-condition in Definition \ref{def_2}.
\end{remark}

\begin{remark}
The asymptotic result requires that the trajectory remains in the dissipative switching set $\mathcal D_\mu$. This assumption is case-wise and explicit: at each time, the active switching case in Eq. \eqref{eq:Psi_piecewise} satisfies the corresponding inequality listed in Eq. \eqref{eq:dissipative_cases}. The complementary set $\mathcal N_\mu$ contains the switching configurations for which strict dissipation cannot be guaranteed. In these configurations, the memristive term may inject a non-dissipative residual into the Lyapunov derivative, and the appropriate conclusion is a practical Lyapunov-type estimate rather than exact asymptotic IOGS.
\end{remark}

\begin{remark}
The adaptive laws Eq. \eqref{eq:update_laws} guarantee boundedness of the parameter-estimation errors and cancellation of the corresponding cross terms in the Lyapunov derivative. However, convergence of $\bar{k},\bar{\sigma},\bar{s},\bar{\xi},\bar{\eta}$ to their target values generally requires additional persistence-of-excitation conditions on the regressors $\Phi_i$. Without such excitation assumptions, the synchronization errors may converge to zero while the parameter estimates converge only to constant values.
\end{remark}

\begin{remark}
The fourth component of the synchronization map is chosen as $\varphi_4(x,y,z,u)=u$. Therefore, the memristive synchronization error is exactly $e_4=u_r-u$. This choice prevents the switchable memristive nonlinearity from being removed by control. Instead, the term $m\big[G(u_r)-G(u)\big]$ is retained explicitly in the closed-loop error dynamics and treated by case-wise Lyapunov estimates. This yields two regimes: exact asymptotic increased-order generalized synchronization when the trajectory remains in the dissipative switching set $\mathcal D_\mu$, and a practical Lyapunov-type estimate with an explicit radius when strict dissipativity is not guaranteed.
\end{remark}

\begin{remark}
The practical Lyapunov-type radius $\delta=\sqrt{8m^2/(\varepsilon\alpha)}$ depends on the memristive strength $m$, the Young parameter $\varepsilon$, and the effective dissipation rate $\alpha=\min\{\lambda_1,\allowbreak\lambda_2,\allowbreak\lambda_3,\allowbreak\lambda_5,\allowbreak 1-\varepsilon/2\}$. Increasing the feedback gains $\lambda_i$ increases the dissipative contribution of the controlled channels, while the factor $1-\varepsilon/2$ measures the remaining dissipation in the memristive channel after the Young estimate has been applied.
\end{remark}


\section{Numerical corroboration of IOGS and discussion}
\label{sec:numerical_results}

To assess the adaptive IOGS scheme numerically, we integrate the full closed-loop
drive--response--adaptation system. Its fourteen state variables
comprise the drive state $X_d\in\mathbb{R}^4$, the response state
$X_r\in\mathbb{R}^5$, and the five adaptive parameter estimates
$\bar{k}$, $\bar{\sigma}$, $\bar{s}$, $\bar{\xi}$, and $\bar{\eta}$.
Combining the drive and response systems
in Eqs.~\eqref{eq:drive_system} and \eqref{eq:response_system}
with the controllers in Eq.~\eqref{eq:controllers_explicit}
and the adaptive laws in Eq.~\eqref{eq:update_laws}
yields

\begin{equation}
\label{eq:full-closed-loop-numerical-system}
\begin{cases}
\begin{aligned}
\dot{x}&=F_1,\quad \dot{y}=F_2,\quad \dot{z}=F_3,\quad \dot{u}=F_4,\\[1mm]
\dot{x}_r&=y_r-a x_r^3+b x_r^2+\bar{k}(h+fu_r^2)x_r\\
&\quad-z_r+I+v_1,\\[1mm]
\dot{y}_r&=c-px_r^2-y_r-\bar{\sigma}w_r+v_2,\\[1mm]
\dot{z}_r&=r\big[\bar{s}(x_r-x_0)-z_r\big]+v_3,\\[1mm]
\dot{u}_r&=-\bar{\xi}u_r+mG(u_r)+x_r+v_4,\\[1mm]
\dot{w}_r&=\bar{\eta}\big[\gamma(y_r-y_0)-\rho w_r\big]+v_5,\\[1mm]
\dot{\bar{k}}&=-\Phi_1e_1,\quad \dot{\bar{\sigma}}=\Phi_2e_2,\quad \dot{\bar{s}}=-\Phi_3e_3,\\[1mm]
\dot{\bar{\xi}}&=\Phi_4e_4,\quad \dot{\bar{\eta}}=-\Phi_5e_5.
\end{aligned}
\end{cases}
\end{equation}

IOGS is assessed through the five components of
$\mathbf{e}(t)=X_r(t)-\varphi(X_d(t))$ and their Euclidean norm
\begin{equation}
\label{eq:error-norm-numerical}
\|\mathbf{e}(t)\|
=
\left(\sum_{i=1}^{5}e_i^2(t)\right)^{1/2}.
\end{equation}
Decay of these quantities toward zero over the simulated interval
provides numerical evidence of asymptotic IOGS.

The numerical results are also interpreted in relation to
the practical Lyapunov-type estimate of
Theorem~\ref{thm:adaptive_practical_lyapunov}. In particular,
we examine the sign of the memristive contribution and
its implications for the evolution of the adaptive
Lyapunov function.

The drive is initialized at
$X_d(0)=(0.1,0.1,0.1,0.1)^\top$, as in
Section \ref{subsec:model_dynamics}. For the synchronization map, we set
$\alpha_1=\alpha_2=\alpha_3=\alpha_4=0.1$ and
$(\beta_0,\ldots,\beta_5)=(0,1,1,0,0.1,0.1)$.
The response initial condition is prescribed by
$X_r(0)=\varphi(X_d(0))+(1,1,1,1,-1)^\top$, ensuring
$\mathbf{e}(0)=(1,1,1,1,-1)^\top$ and giving
$X_r(0)\approx(1.1100,1.1020,1.1010,1.1000,-0.6995)^\top$.
All five adaptive parameter estimates are initialized at zero,
and the feedback gains are fixed at
$\lambda_1=\lambda_2=\lambda_3=\lambda_5=2$.
We integrate Eq.~\eqref{eq:full-closed-loop-numerical-system}
over $t\in[0,5\times10^5]$ using the numerical method and
tolerances specified in Section \ref{subsec:model_dynamics}.
The results are displayed on logarithmic time axes to resolve
both the initial transient and the subsequent slow relaxation.

\begin{figure}
\centering
\includegraphics[width=\columnwidth]{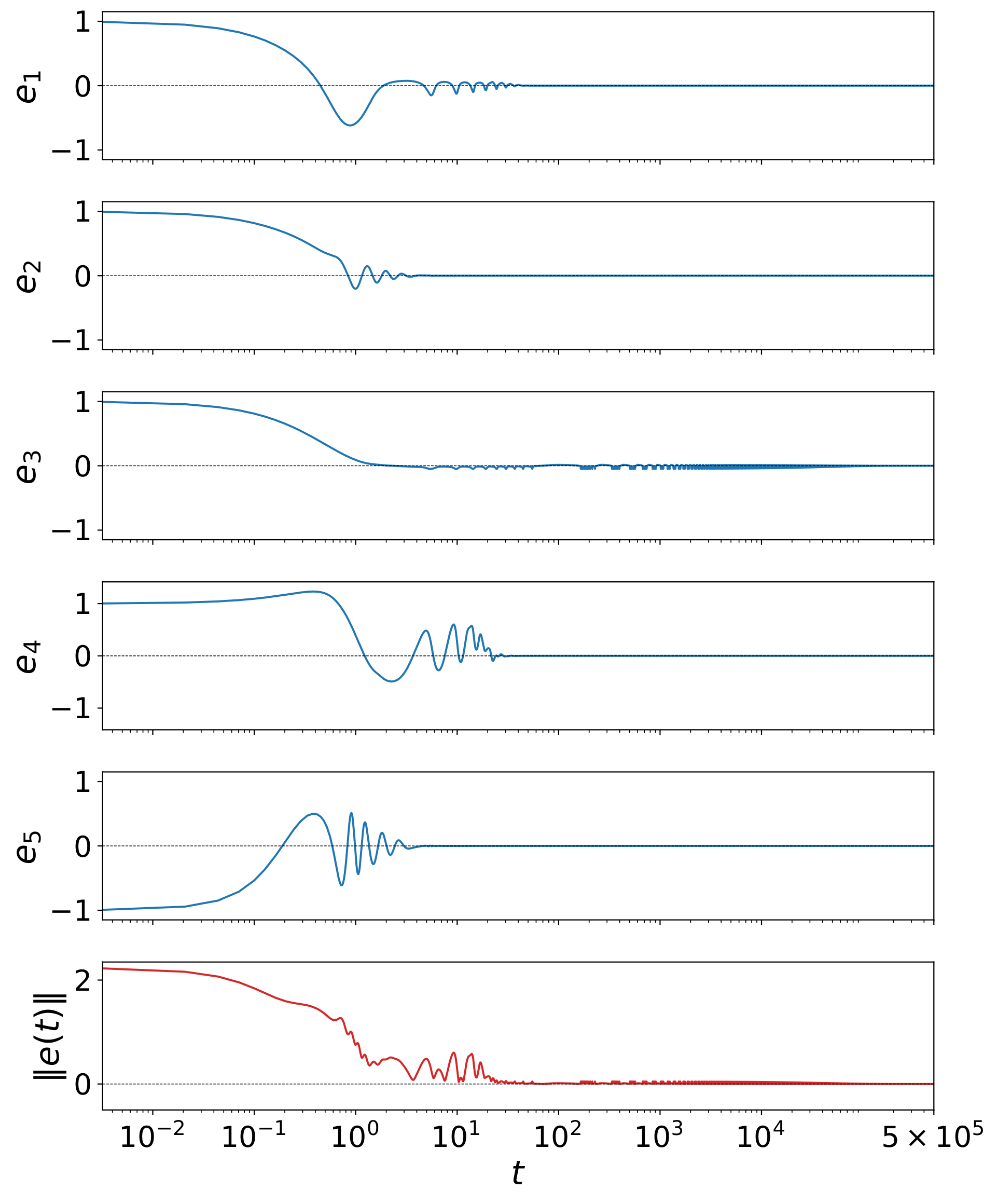}
\caption{Time evolution of the synchronization errors $e_1(t),\ldots,e_5(t)$ and their Euclidean norm $\|e(t)\|$. Logarithmic time axes resolve the transient oscillations and the subsequent relaxation toward zero.}
\label{fig:sync_errors}
\end{figure}

Figure~\ref{fig:sync_errors} shows the synchronization errors and
their Euclidean norm. Following transient overshoots and oscillations, all five
components approach zero, with $e_3$ exhibiting a longer
relaxation than the other components. The error norm decreases
nonmonotonically, reaching approximately $6.19\times10^{-7}$ at
the final integration time. Its maximum observed value over the
final $20\%$ of the simulated interval is $1.22\times10^{-5}$.
These results are consistent with asymptotic IOGS for the
trajectory considered.

Figure~\ref{fig:parameter_errors} shows the parameter-estimation
errors $\widehat{k}(t)$, $\widehat{\sigma}(t)$, $\widehat{s}(t)$,
$\widehat{\xi}(t)$, and $\widehat{\eta}(t)$, defined previously
as the differences between the adaptive estimates and their
nominal target values. All five errors approach zero over the
simulated interval, although the relaxation of $\widehat{s}$
is substantially slower than that of the other components.
At the final integration time,
$\widehat{s}\approx-1.53\times10^{-4}$, while the absolute
values of the remaining parameter-estimation errors are below
$3\times10^{-9}$. These values indicate accurate recovery of
the nominal parameters for the computed trajectory.
Synchronization alone, however, does not imply parameter
identification. A theoretical guarantee of convergence of the
adaptive estimates to their targets requires additional analysis, typically involving excitation conditions on the regressors; such conditions, beyond the scope of the current study, are not established here.

\begin{figure}
\centering
\includegraphics[width=\columnwidth]{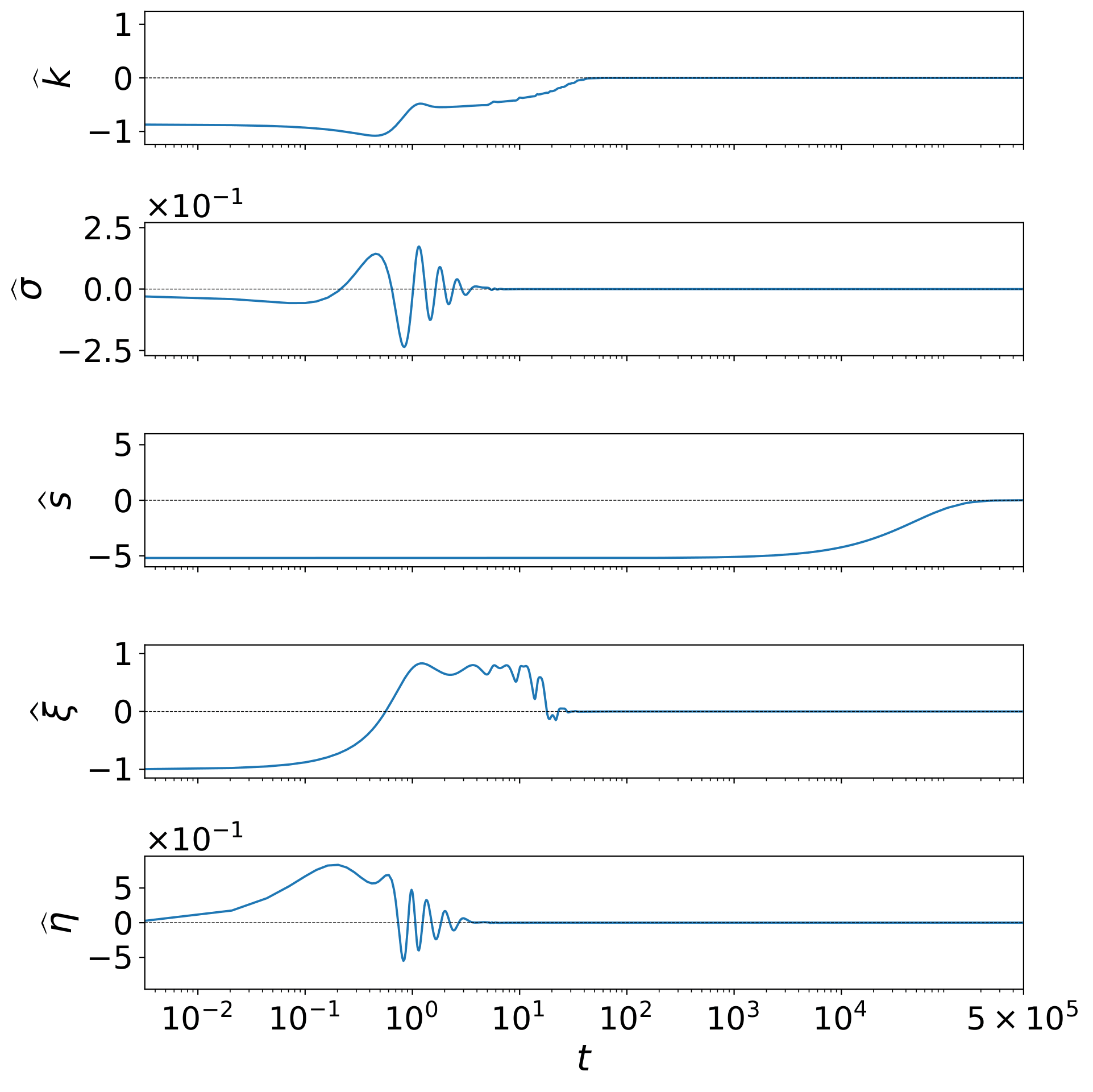}
\caption{Time evolution of the parameter-estimation errors
$\widehat{k}(t)$, $\widehat{\sigma}(t)$, $\widehat{s}(t)$,
$\widehat{\xi}(t)$, and $\widehat{\eta}(t)$ relative to the
nominal target values. Logarithmic time axes reveal the
distinct relaxation times, with $\widehat{s}(t)$ approaching
zero substantially more slowly than the other components.}
\label{fig:parameter_errors}
\end{figure}

We next examine the memristive contribution to the Lyapunov
derivative along the computed trajectory.
Figure~\ref{fig:dissipative_check} shows the previously defined
quantities $\Psi(t)$ and $\Theta(t)=\Psi(t)/e_4^2(t)$.
The normalized coefficient is evaluated only where $e_4(t)$
remains numerically resolved, since division by $e_4^2(t)$
becomes sensitive to numerical errors as $e_4(t)$ approaches zero.
The uniform dissipativity assumption of
Theorem~\ref{thm:adaptive_asymptotic} requires a fixed
$\mu>0$ such that
\[
\Psi(t)\le-\mu e_4^2(t)
\qquad\text{for all }t\ge0,
\]
or, equivalently, $\Theta(t)\le-\mu$ whenever $e_4(t)\ne0$.

For the chosen value $m=0.5$, monotonicity and the
$2$-Lipschitz continuity of $G$ imply
\[
\Theta(t)
=
-1+\frac{1}{2}
\frac{G(u_r(t))-G(u(t))}{u_r(t)-u(t)}
\in[-1,0],
\enspace e_4(t)\ne0.
\]
Thus, the memristive contribution is nonpositive throughout
the trajectory. In particular, $\Theta=-1$ when both
memristive states occupy the same outer branch
($\mathcal{C}_1$ or $\mathcal{C}_9$), whereas $\Theta=0$
when both occupy the central branch ($\mathcal{C}_5$).

As shown in Figure~\ref{fig:dissipative_check}, the trajectory
repeatedly enters $\mathcal{C}_5$ while $e_4$ remains nonzero
and numerically resolved. These visits preclude any fixed
$\mu>0$ satisfying the uniform dissipativity condition for
all $t\ge0$. Consequently,
Theorem~\ref{thm:adaptive_asymptotic} does not directly
establish asymptotic IOGS for this particular trajectory, despite the
observed decay of the synchronization errors.
Thus, the simulation provides numerical evidence of synchronization
beyond the \textit{sufficient} conditions of
Theorem~\ref{thm:adaptive_asymptotic}, rather than a numerical
corroboration of that theorem under its stated hypotheses.

The switching-independent estimate of
Theorem~\ref{thm:adaptive_practical_lyapunov} remains applicable
under its stated assumptions. It provides a Lyapunov decrease
condition outside the prescribed error neighborhood, but does
not by itself establish ultimate boundedness or asymptotic
convergence.  For the present parameter choice, the stronger sign property
$\Psi\le0$ implies $\dot{V}\le0$ and hence boundedness of the
synchronization and parameter-estimation errors. This
nonincrease of $V$, however, does not by itself imply
$e_4(t)\to0$. The observed decay of the memristive error
is therefore a numerical finding beyond the convergence
conclusions established by the present Lyapunov analysis,
and is compatible with the estimate of
Theorem~\ref{thm:adaptive_practical_lyapunov}.

\begin{figure}
\centering
\includegraphics[width=\columnwidth]{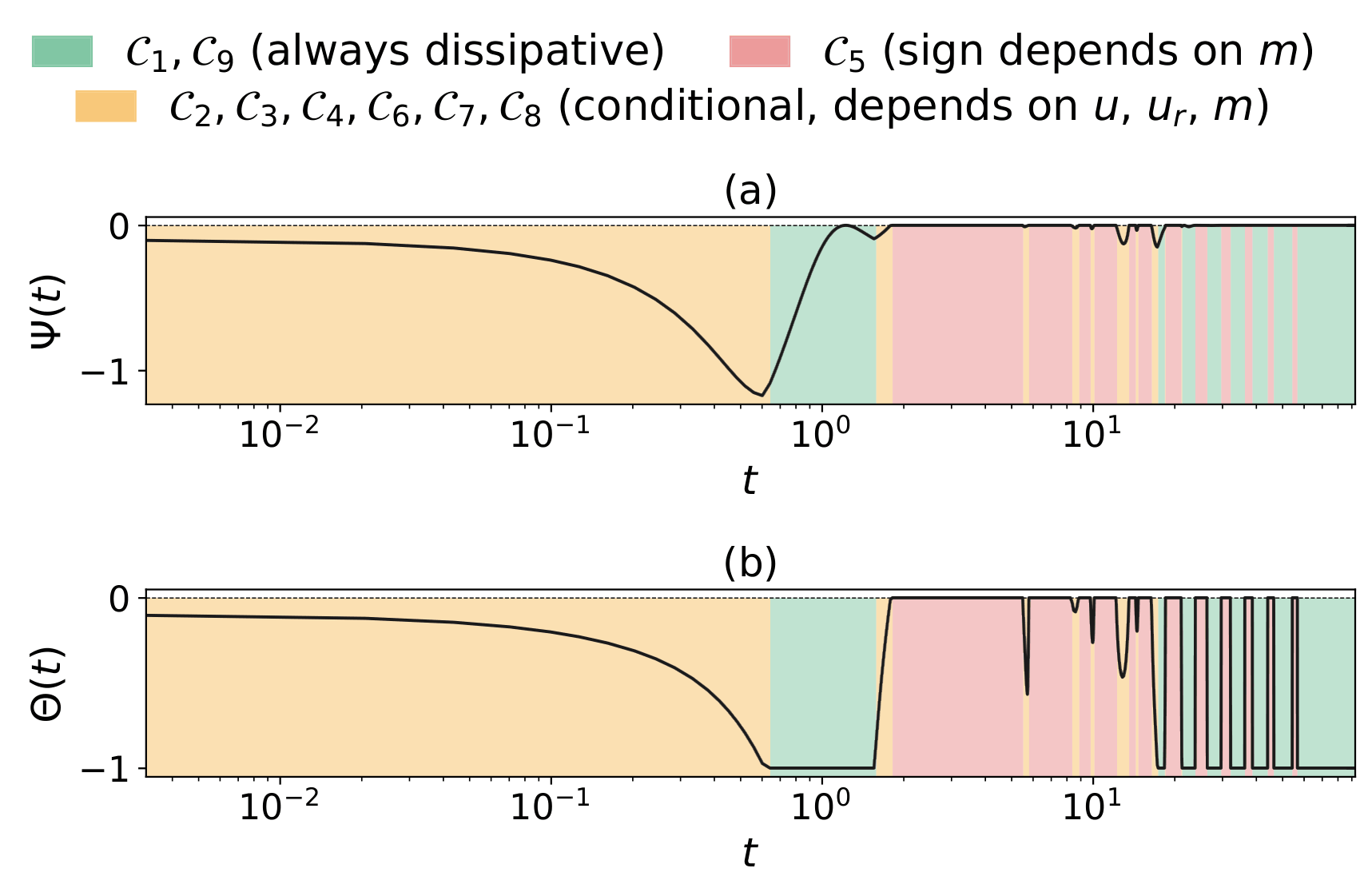}
\caption{Memristive dissipation along the computed trajectory
for $m=0.5$: (a) Lyapunov contribution $\Psi(t)$ and
(b) normalized coefficient $\Theta(t)$.
The displayed interval is restricted to times for which
$e_4(t)$ remains numerically resolved.
Background shading identifies the active switching case,
using the convention of Figure~\ref{fig:switching_regimes}.
Visits to the central region $\mathcal{C}_5$ yield
$\Theta=0$ for nonzero $e_4$, so the uniform strict
dissipativity condition is not satisfied, although
$\Psi\le0$ throughout.}
\label{fig:dissipative_check}
\end{figure}

\section{The IOGS problem: a data-driven approach}
\label{sec:data-driven}

In this section, we tackle the IOGS problem through the lens of machine learning. For a complex system, accurate governing equations may be unavailable, with only measurements of its activity—such as membrane potentials and ionic currents in neural systems—accessible. Data-driven approaches therefore offer a means of designing control schemes directly from observed dynamics.

In this paper, we address IOGS under the assumption that the full states of both systems are sampled at interval $\Delta t$ and that the control input enters the response additively, as in Eq. \eqref{eq:response_system}. From uncoupled state trajectories, Next Generation Reservoir Computing (NG-RC) \cite{Gauthier.2021} learns one-step predictors of the drive and response dynamics (Section \ref{subsec:ngrc}). We then replace the prescribed analytical map $\varphi$ in Eq. \eqref{eq:nonlinear_phi} with a parameterized map $\varphi_\theta$, learned through a self-supervised objective that promotes consistency between the two predictors without requiring paired drive--response observations (Section \ref{subsec:map}). The learned predictors and map define a sampled-data coupling law (Section \ref{subsec:coupling}). The governing equations are used only to generate the recordings and simulate the controlled dynamics; they are not supplied to the learning or control algorithms. The complete procedure, comprising offline learning of the predictors
and synchronization map followed by online sampled-data control,
is summarized in Algorithm~\ref{alg:data_driven_iogs} in the Appendix \ref{appendix}.

\subsection{Next generation reservoir computing}
\label{subsec:ngrc}

Our goal is to learn approximations $\hat{F}_d$ and $\hat{F}_r$ of the drive vector field $F_d$ in Eq. \eqref{eq:drive_system} and the uncontrolled response vector field $F_r$, defined by the right-hand side of Eq. \eqref{eq:response_system} with $\mathbf{v}=\mathbf{0}$, from sampled trajectories alone. Here, hats denote learned quantities rather than parameter errors as in Eq. \eqref{eq:param_errors}. To this end, we use NG-RC \cite{Gauthier.2021}, which replaces the randomly connected recurrent reservoir of an echo state network \cite{Jaeger.2001} with an explicit feature vector comprising state variables, optional time delays, and their monomials (Figure \ref{fig:ngrc}).

\begin{figure*}
\centering
\includegraphics[width=0.85\textwidth]{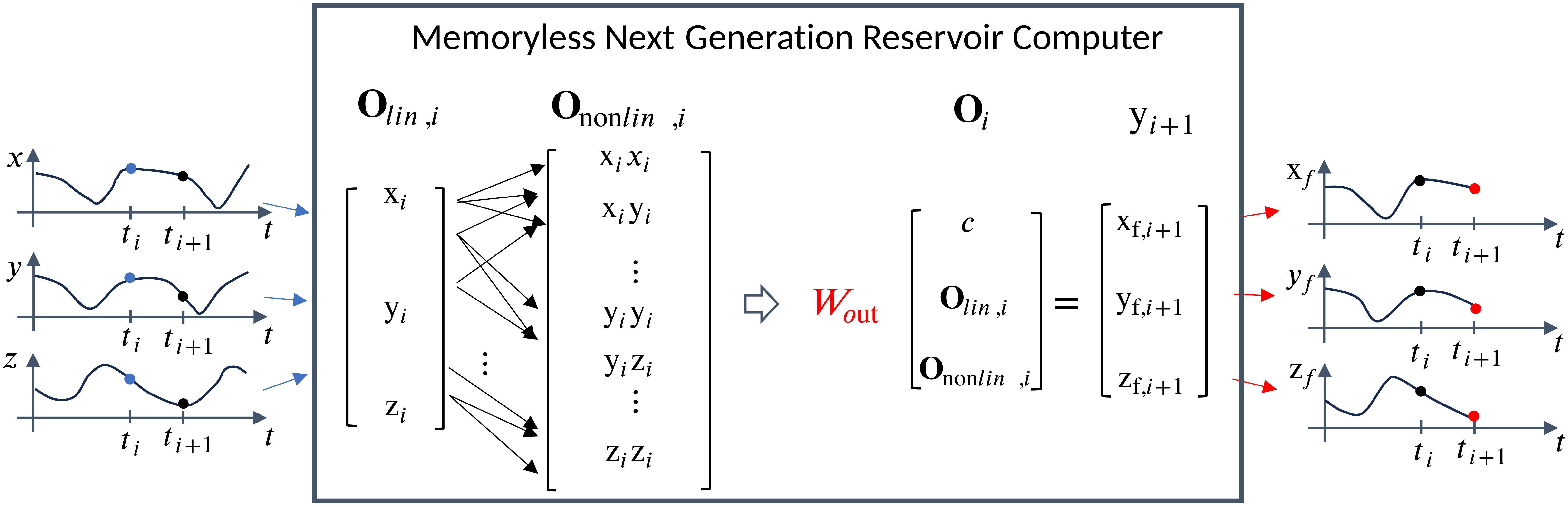}
\caption{Schematic of memoryless NG-RC, adapted from \cite{Gauthier.2021}. The current state is represented by a feature vector comprising a constant, linear state components, and nonlinear monomials. A readout matrix $W_{\mathrm{out}}$, fitted by ridge regression, maps these features to estimated state-increment rates, which determine the one-step prediction. No delayed states or recurrent reservoir are used.}
\label{fig:ngrc}
\end{figure*}

Training and assessment data are generated by simulating the two systems independently, with the response parameters fixed at the nominal values specified in Section \ref{subsec:parameters}, namely $(\bar{k},\bar{\sigma},\bar{s},\bar{\xi},\bar{\eta})=(k,\sigma,s,\xi,\eta)$. Trajectories are sampled at interval $\Delta t=0.02$. We write $\mathbf{X}(t)\in\mathbb{R}^n$ for either sampled state, with $n=4$ for the drive and $n=5$ for the response. Recordings from distinct initial conditions are partitioned into disjoint training, evaluation, and test sets, assigning each recording entirely to one set. The training set is used to fit the readout matrix $\mathbf{W}_{\mathrm{out}}$, the evaluation set to select hyperparameters, and the held-out test set to assess generalization to unseen trajectories. States are standardized as $\tilde{\mathbf{X}}(t)=\mathbf{D}^{-1}(\mathbf{X}(t)-\bar{\mathbf{X}})$, where $\bar{\mathbf{X}}$ contains the componentwise training means and $\mathbf{D}=\operatorname{diag}(d_1,\ldots,d_n)$ contains the corresponding standard deviations, assumed positive. These training statistics are also used to standardize the evaluation and test data.

Unlike \cite{Gauthier.2021}, and following \cite{Kent.2024}, we use a memoryless formulation. For the deterministic systems considered here, the full state determines the subsequent evolution, so delay coordinates are not required. This also permits evaluating $\hat{F}_r$ at mapped drive states (Section \ref{subsec:map}) and controlled response states (Section \ref{subsec:coupling}), for which an uncontrolled response history is unavailable. The nonlinear features $\mathbf{O}_{\mathrm{nonlin}}(t)$ comprise all distinct monomials of total degree $2$ through $q$ in the standardized state components, giving
\begin{equation}
\mathbf{O}(t)=\big[1;\ \tilde{\mathbf{X}}(t);\ \mathbf{O}_{\mathrm{nonlin}}(t)\big].
\label{eq:ngrc_features}
\end{equation}

Following \cite{Gauthier.2021}, we train the readout on consecutive state increments, scaled here to standardized finite-difference rates $\mathbf{Y}(t)=\big(\tilde{\mathbf{X}}(t+\Delta t)-\tilde{\mathbf{X}}(t)\big)/\Delta t$. Stacking features and targets column-wise over the training recordings yields the matrices $\mathbf{O}$ and $\mathbf{Y}$. Ridge regression then gives
\begin{equation}
\mathbf{W}_{\mathrm{out}}=\mathbf{Y}\mathbf{O}^\top\big(\mathbf{O}\mathbf{O}^\top+\nu\mathbf{I}\big)^{-1},
\label{eq:ngrc_ridge}
\end{equation}
where $\nu>0$ is the Tikhonov regularization parameter.

In the original coordinates, $\hat{F}(\mathbf{X}):=\mathbf{D}\mathbf{W}_{\mathrm{out}}\mathbf{O}(\mathbf{X})$ approximates the finite-time increment rate $\big(\Phi_{\Delta t}(\mathbf{X})-\mathbf{X}\big)/\Delta t$, where $\Phi_{\Delta t}$ denotes the uncontrolled flow map and $\mathbf{O}(\mathbf{X})$ is constructed from the standardized state. For locally Lipschitz $F$, this rate equals $F(\mathbf{X})+\mathcal{O}(\Delta t)$ as $\Delta t\to0$, distinguishing finite-step bias from regression error. The corresponding one-step predictor is
\begin{equation}
\hat{\mathbf{X}}(t+\Delta t)=\mathbf{X}(t)+\Delta t\,\hat{F}(\mathbf{X}(t)).
\label{eq:ngrc_step}
\end{equation}
Fitting Eq. \eqref{eq:ngrc_ridge} separately to the drive and response recordings yields $\hat{F}_d$ and $\hat{F}_r$. Autonomous forecasts are obtained by iterating Eq. \eqref{eq:ngrc_step}.

The hyperparameters $(q,\nu)$ are selected separately for each system by Bayesian optimization using a Gaussian-process surrogate and the expected-improvement acquisition function \cite{Snoek.2012}, implemented in \texttt{scikit-optimize} \cite{Head.2021}.
Each search comprises 30 evaluations, including 8 initial random trials. Since map learning and control use one-step predictions, we optimize short-horizon accuracy rather than long-term autonomous forecasting. The objective is the standardized RMSE over $N_s$ autonomous rollouts, each initialized at a measured evaluation state $\mathbf{X}(t_s)$ and advanced for $N_p$ steps:
\begin{multline}
\mathrm{RMSE}=\\*
\sqrt{\frac{1}{nN_sN_p}\sum_{s=1}^{N_s}\sum_{j=1}^{N_p}\big\|\hat{\tilde{\mathbf{X}}}(t_s+j\Delta t)-\tilde{\mathbf{X}}(t_s+j\Delta t)\big\|_2^2},
\label{eq:ngrc_rmse}
\end{multline}
where $N_p=50$ spans one time unit and $N_s=2000$, with 1000 initial times sampled randomly from each of the two evaluation recordings.

We restrict $q\leq3$ to obtain a compact feature basis that contains all polynomial terms of the HR vector fields. The switching term $G$ is not represented exactly by this basis; moreover, the finite-time increment rates need not be cubic even for a cubic vector field. We search $\nu\in[10^{-10},1]$ on a logarithmic scale. Both searches select $q=3$, yielding 35 drive features and 56 response features, including the constant term. On the test recordings, Eq. \eqref{eq:ngrc_rmse} gives RMSEs of $3.7\times10^{-3}$ and $3.3\times10^{-3}$ for the drive and response predictors, respectively. Figure \ref{fig:ngrc_results} compares longer autonomous forecasts with the corresponding test trajectories.
\begin{figure*}
\centering
\includegraphics[width=0.71\textwidth]{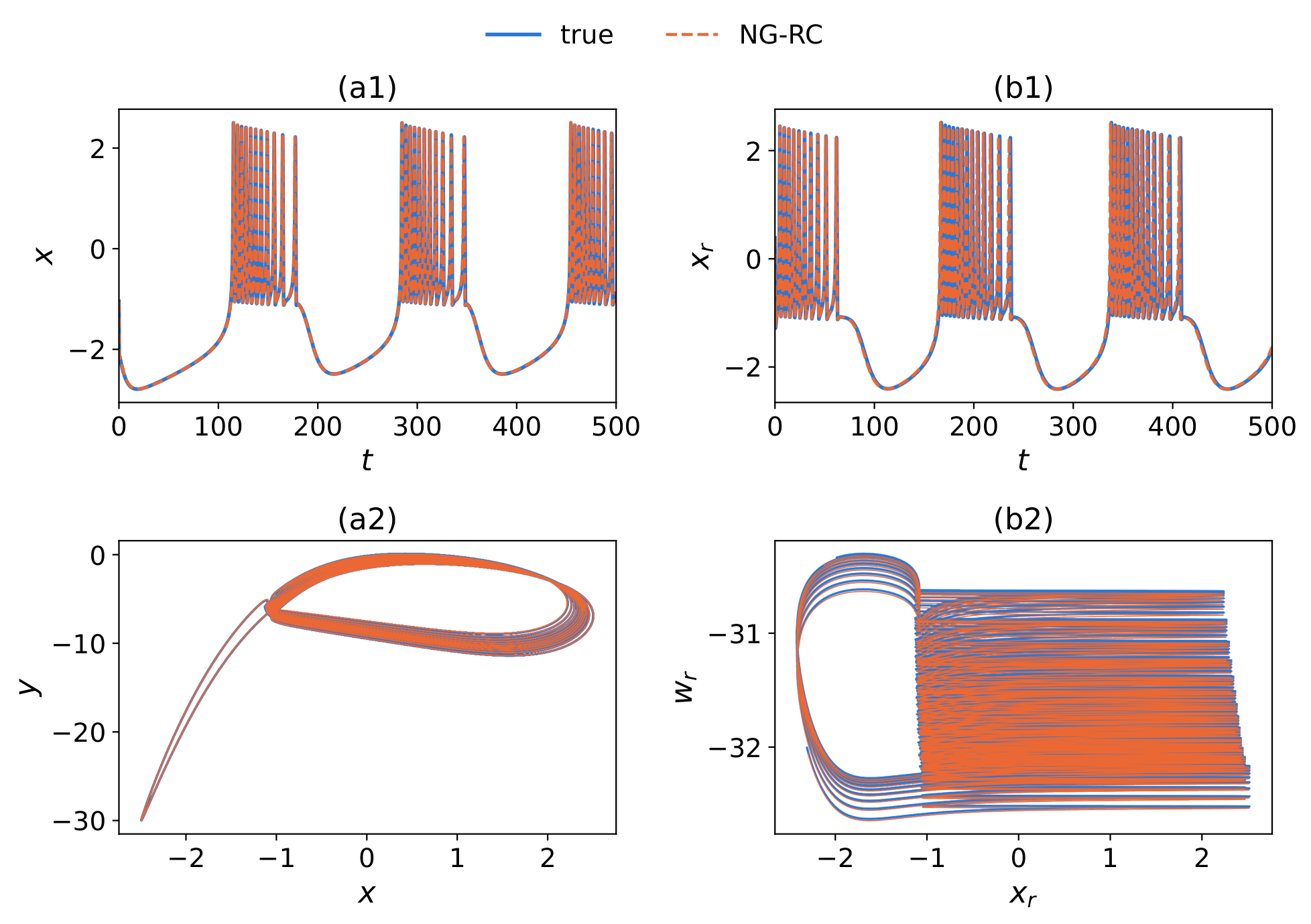}
\caption{Autonomous NG-RC forecasts compared with test trajectories for (a) the drive and (b) the response: (a1, b1) membrane potentials; (a2, b2) phase-space projections for $t\geq300$.}
\label{fig:ngrc_results}
\end{figure*}

\subsection{Self-supervised learning of the synchronization map}
\label{subsec:map}

We seek a map $\varphi_\theta:\mathbb{R}^4\to\mathbb{R}^5$ whose graph is approximately invariant under the learned one-step dynamics, with the aim of reducing the control required to maintain $\mathbf{X}_r=\varphi_\theta(\mathbf{X}_d)$. The uncoupled recordings provide no paired drive--response targets for supervised regression. We therefore learn $\varphi_\theta$ through a self-supervised consistency objective: advancing a drive state with its learned predictor and then mapping should approximately agree with mapping first and advancing with the learned response predictor.

We parameterize $\varphi_\theta$ as an affine map in the standardized coordinates of Section \ref{subsec:ngrc}. Let $\bar{\mathbf{X}}_d,\bar{\mathbf{X}}_r$ be the training means and $\mathbf{D}_d,\mathbf{D}_r$ the diagonal matrices of training standard deviations. With $\tilde{\mathbf{X}}_d=\mathbf{D}_d^{-1}(\mathbf{X}_d-\bar{\mathbf{X}}_d)$, we define
\begin{equation}
\varphi_\theta(\mathbf{X}_d)=\bar{\mathbf{X}}_r+\mathbf{D}_r\big(\mathbf{A}\tilde{\mathbf{X}}_d+\mathbf{b}\big),
\label{eq:learned_map}
\end{equation}
where $\theta=(\mathbf{A},\mathbf{b})$, $\mathbf{A}\in\mathbb{R}^{5\times4}$, and $\mathbf{b}\in\mathbb{R}^5$, giving 25 trainable parameters. This family includes the reference correspondence $\varphi_{\theta_0}(\mathbf{X}_d)=(x,y,z,u,\bar{w}_r)^\top$, which matches the four shared variables and fixes the fifth at its response training mean $\bar{w}_r$. Standardization makes the entries of $\mathbf{A}$ dimensionless and comparable across components. Unlike the nonlinear map in Eq. \eqref{eq:nonlinear_phi}, this family is affine and does not impose $\varphi_{\theta,4}=u$; hence, it does not preserve the specific memristive error structure used in Section \ref{subsec:casewise}.

We fit $\theta$ using $N$ drive states $\mathbf{X}_{d,n}$ sampled from the drive training recordings. The two compositions of the map and the one-step predictors in Eq. \eqref{eq:ngrc_step} give
\begin{equation}
\begin{split}
\mathbf{a}_n(\theta)&=\varphi_\theta\big(\mathbf{X}_{d,n}+\Delta t\,\hat{F}_d(\mathbf{X}_{d,n})\big),
\\
\mathbf{b}_n(\theta)&=\varphi_\theta(\mathbf{X}_{d,n})+\Delta t\,\hat{F}_r\big(\varphi_\theta(\mathbf{X}_{d,n})\big),
\end{split}
\label{eq:two_paths}
\end{equation}
where $\mathbf{a}_n$ advances then maps, whereas $\mathbf{b}_n$ maps then advances. Their agreement supplies a self-supervised training signal without paired response targets. We define the dynamics-consistency loss as
\begin{equation}
L_{\mathrm{dyn}}(\theta)=\frac{1}{N}\sum_{n=1}^{N}\left\|\frac{\mathbf{a}_n(\theta)-\mathbf{b}_n(\theta)}{\Delta t}\right\|_2^2.
\label{eq:L_dyn}
\end{equation}

The consistency loss alone admits collapsed solutions: any constant map $\varphi_\theta\equiv\mathbf{c}$ satisfying $\hat{F}_r(\mathbf{c})=\mathbf{0}$ yields $L_{\mathrm{dyn}}=0$ and ignores the drive. Such a solution targets an equilibrium of the learned response dynamics rather than a varying response to the drive. To discourage collapse, we introduce variance and covariance penalties inspired by VICReg \cite{Bardes.2022}, applied to the first four standardized target components. Let $\mathbf{A}_{1:4}$ contain the first four rows of $\mathbf{A}$, $\mathbf{C}_d$ be the sample covariance of the standardized training inputs $\tilde{\mathbf{X}}_{d,n}$, and $\mathbf{R}$ be the sample covariance of the first four standardized response components in the training recordings. The corresponding target covariance is $\mathbf{S}(\theta)=\mathbf{A}_{1:4}\mathbf{C}_d\mathbf{A}_{1:4}^\top$. We define
\begin{equation}
\begin{split}
L_{\mathrm{var}}(\theta)&=\sum_{i=1}^{4}\max\Big(0,\,1-\sqrt{S_{ii}(\theta)}\Big)^2,
\\
L_{\mathrm{cov}}(\theta)&=\sum_{\substack{i,j=1\\i\neq j}}^{4}\big(S_{ij}(\theta)-R_{ij}\big)^2.
\end{split}
\label{eq:collapse_terms}
\end{equation}

The variance penalty discourages target standard deviations below the unit scale of the standardized response data, while the covariance penalty discourages discrepancies in their off-diagonal covariances. Unlike VICReg, which penalizes off-diagonal covariances toward zero, we retain the dependence structure observed in the response recordings. These soft penalties do not prescribe a coordinate correspondence or guarantee a nondegenerate map. The fifth target is exempt from both penalties, allowing its variability to adjust to the slow $w_r$ dynamics without imposing the variance of the uncontrolled response. Its dynamical compatibility remains enforced through $L_{\mathrm{dyn}}$. The total training objective is
\begin{equation}
L(\theta)=L_{\mathrm{dyn}}(\theta)+\lambda_{\mathrm{var}}L_{\mathrm{var}}(\theta)+\lambda_{\mathrm{cov}}L_{\mathrm{cov}}(\theta),
\label{eq:map_loss}
\end{equation}
with $\lambda_{\mathrm{var}}=\lambda_{\mathrm{cov}}=1$.

We minimize Eq. \eqref{eq:map_loss} using $N=14400$ drive states, with 1800 sampled randomly from each of the eight training recordings. Optimization uses L-BFGS-B \cite{Byrd.1995}, implemented in \texttt{SciPy} \cite{Virtanen.2020}, without parameter bounds. The initialization $\mathbf{A}=\mathbf{0}$, $\mathbf{b}=\mathbf{0}$ maps every drive state to the response training mean $\bar{\mathbf{X}}_r$, imposing no coordinate correspondence. In the implementation, $\sqrt{S_{ii}(\theta)}$ in Eq. \eqref{eq:collapse_terms} is evaluated as $\sqrt{S_{ii}(\theta)+\epsilon}$ with $\epsilon=10^{-12}$, which makes $L$ differentiable at the initialization, and its exact gradient is supplied to the optimizer. The stopping tolerances are $10^{-10}$ for the relative objective reduction and $10^{-7}$ for the infinity norm of the supplied gradient.

The objective decreases from $75.6$ to $8.3\times10^{-3}$ over 763 iterations (Figure \ref{fig:sync_map_training}). At the final iterate, $L_{\mathrm{dyn}}=4.2\times10^{-3}$ on the training states and $3.2\times10^{-3}$ on the evaluation states. As an independent consistency check, we replace both learned one-step predictors in Eq. \eqref{eq:two_paths} with fourth-order Runge--Kutta (RK4) steps of the uncoupled governing equations. Evaluated on the same training states, this diagnostic gives $L_{\mathrm{dyn}}=3.3\times10^{-3}$ at the learned parameters (dashed curve in Figure \ref{fig:sync_map_training}). This comparable residual supports consistency with the underlying dynamics on the sampled states, although it does not exclude predictor errors elsewhere.

Without a prescribed coordinate correspondence, the learned linear part approximates that of the reference map $\varphi_{\theta_0}$. The matrix $\mathbf{A}_{1:4}$ is nearly diagonal: its diagonal entries lie in $[0.98,1.01]$ and differ from the corresponding reference entries by at most $0.04$, while its off-diagonal entries have magnitude below $0.1$. Its singular values lie in $[0.95,1.05]$, indicating a well-conditioned mapping of the four standardized drive coordinates. The offsets differ from the reference values, most notably in the $z_r$ component, whose target exceeds $z$ by $0.63$ at the mean drive state. The fifth row is dominated by a coefficient of $0.027$ on $\tilde{z}$, yielding a weakly varying $w_r$ target consistent with its slow dynamics.

\begin{figure}
\centering
\includegraphics[width=\columnwidth]{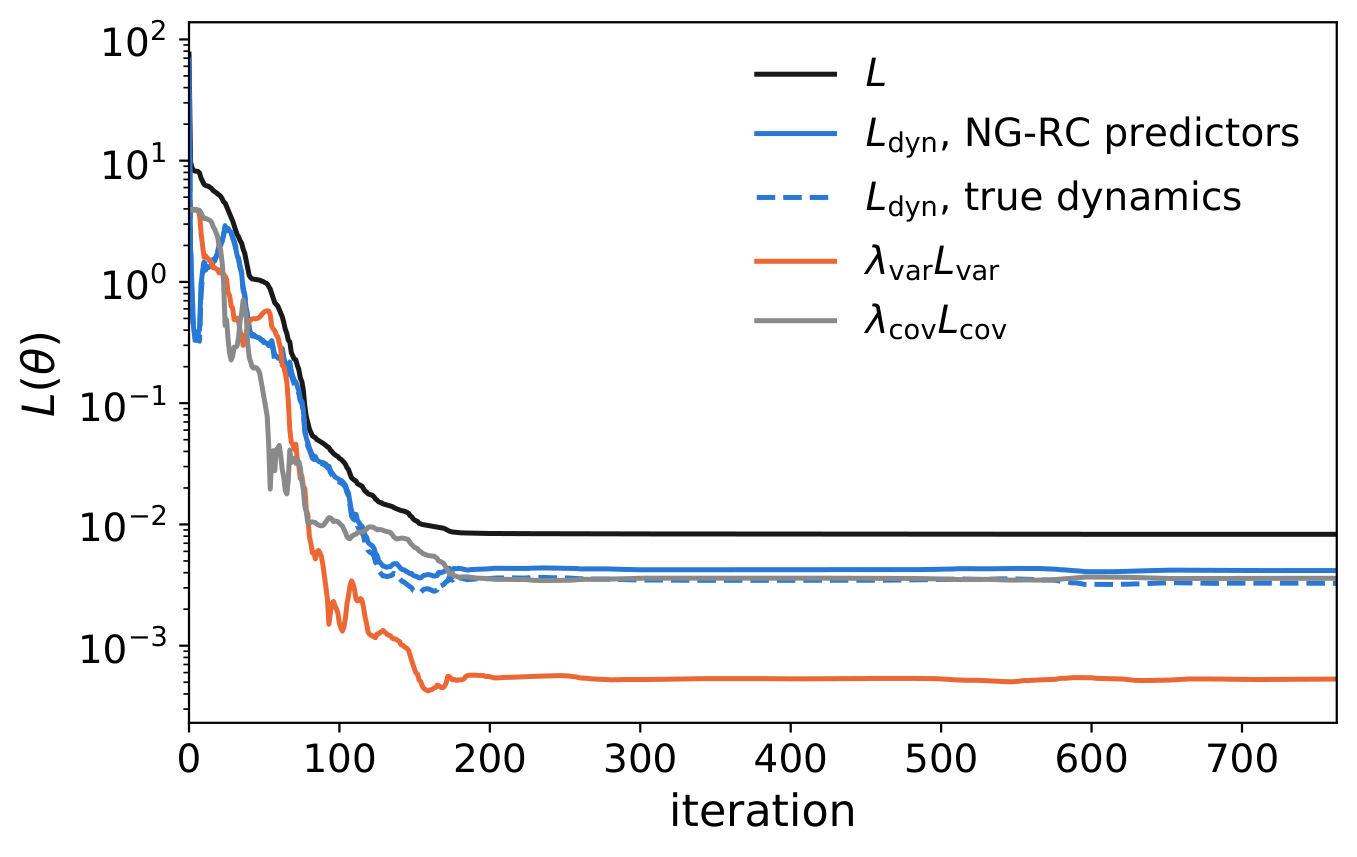}
\caption{Training objective $L$ in Eq. \eqref{eq:map_loss} and its weighted components during L-BFGS-B optimization. The dashed curve shows $L_{\mathrm{dyn}}$ evaluated using RK4 steps of the governing equations.}
\label{fig:sync_map_training}
\end{figure}

\subsection{Data-driven coupling law}
\label{subsec:coupling}

For the learned map, we define the synchronization error analogously to Eq. \eqref{eq:error_definition} as $\mathbf{e}_\theta(t)=\mathbf{X}_r(t)-\varphi_\theta(\mathbf{X}_d(t))$. The response is actuated through all five channels of Eq. \eqref{eq:response_system}. Inputs are updated at sampling times and applied through a zero-order hold: $\mathbf{v}(t')=\mathbf{v}(t)$ for $t'\in[t,t+\Delta t)$.

We construct a sampled-data analogue of Eq. \eqref{eq:controllers_explicit} using learned predictors, following the NG-RC control approach of \cite{Kent.2024} and the discrete-time neural-network control framework of \cite{Sarangapani.2006}. Let $\hat{\mathbf{X}}_d(t+\Delta t)=\mathbf{X}_d(t)+\Delta t\,\hat{F}_d(\mathbf{X}_d(t))$ and $\hat{\mathbf{X}}_r^{\mathrm{free}}(t+\Delta t)=\mathbf{X}_r(t)+\Delta t\,\hat{F}_r(\mathbf{X}_r(t))$ denote the uncoupled one-step predictions. Approximating the effect of the held input to first order gives the controlled predictor $\hat{\mathbf{X}}_r^{\mathrm{free}}(t+\Delta t)+\Delta t\,\mathbf{v}(t)$. Requiring its synchronization error to equal $\zeta\,\mathbf{e}_\theta(t)$, where $\zeta=e^{-\kappa\Delta t}\in(0,1)$ and $\kappa>0$, yields
\begin{equation}
\begin{split}
&\mathbf{v}(t)=\\
&\frac{\varphi_\theta\big(\hat{\mathbf{X}}_d(t+\Delta t)\big)-\hat{\mathbf{X}}_r^{\mathrm{free}}(t+\Delta t)+\zeta\,\mathbf{e}_\theta(t)}{\Delta t}.
\end{split}
\label{eq:coupling_law}
\end{equation}
This contraction is imposed on the predictor; prediction errors and the approximate treatment of the held input prevent it from being an exact identity for the physical systems.

Substituting $\mathbf{X}_r(t)=\varphi_\theta(\mathbf{X}_d(t))+\mathbf{e}_\theta(t)$ into Eq. \eqref{eq:coupling_law} gives
\begin{equation}
\begin{split}
\mathbf{v}(t)={}&\frac{\varphi_\theta\big(\hat{\mathbf{X}}_d(t+\Delta t)\big)-\varphi_\theta(\mathbf{X}_d(t))}{\Delta t}\\
&-\hat{F}_r(\mathbf{X}_r(t))-\frac{1-\zeta}{\Delta t}\,\mathbf{e}_\theta(t).
\end{split}
\label{eq:coupling_law_split}
\end{equation}
The three terms provide target-rate prediction, response-drift compensation, and proportional error feedback, respectively, paralleling Eq. \eqref{eq:controllers_explicit}. The effective feedback gain $(1-\zeta)/\Delta t$ tends to $\kappa$ as $\Delta t\to0$. Unlike the analytical controller, this law treats the memristive channel identically to the others: its dynamics are approximately compensated through the learned predictors rather than retaining an explicit switching difference $G(u_r)-G(u)$. The response parameters remain fixed, and no adaptive parameter laws are used.

For numerical assessment, the governing equations represent the physical drive and response, while the controller accesses only sampled states. Unlike \cite{Kobiolka.2025}, where synchronization is imposed between two reservoir observers, the input acts directly on the response system. Both systems evolve uncoupled until $t=100$, when control is activated with $\kappa=1$. Between sampling times, we integrate both systems using RK4 with four substeps of size $\Delta t/4$, holding the response input constant. Figure \ref{fig:ml_closed_loop} shows a representative trajectory: all five response components approach their targets within a few time units and subsequently track the mapped drive dynamics with small residual errors.

The errors in Figure \ref{fig:ml_sync_errors} initially decay approximately at the prescribed rate $e^{-\kappa(t-100)}$. From $t\approx105$, the error norm fluctuates between approximately $10^{-3}$ and $10^{-2}$. We quantify tracking accuracy by the time average of $\|\mathbf{D}_r^{-1}\mathbf{e}_\theta(t)\|_2/\sqrt{5}$ and control effort by the temporal RMS of $\|\mathbf{D}_r^{-1}\mathbf{v}(t)\|_2$, both evaluated over the final fifty time units. Across eight initial-condition pairs, the median tracking error and control effort are $1.4\times10^{-3}$ and $2.1\times10^{-2}$, respectively. Using the reference map $\varphi_{\theta_0}$ increases these values to $4.8\times10^{-3}$ and $1.7\times10^{-1}$. Thus, the learned map reduces median control effort by a factor of approximately eight while improving tracking accuracy. Replacing both learned predictors with RK4 steps of the governing equations reduces the median error to $2.3\times10^{-4}$, indicating a substantial contribution from predictor inaccuracies. The remaining discrepancy is consistent with the approximate sampled-data control model and numerical integration; exact synchronization is not established by these experiments.

\begin{figure}
\centering
\includegraphics[width=\columnwidth]{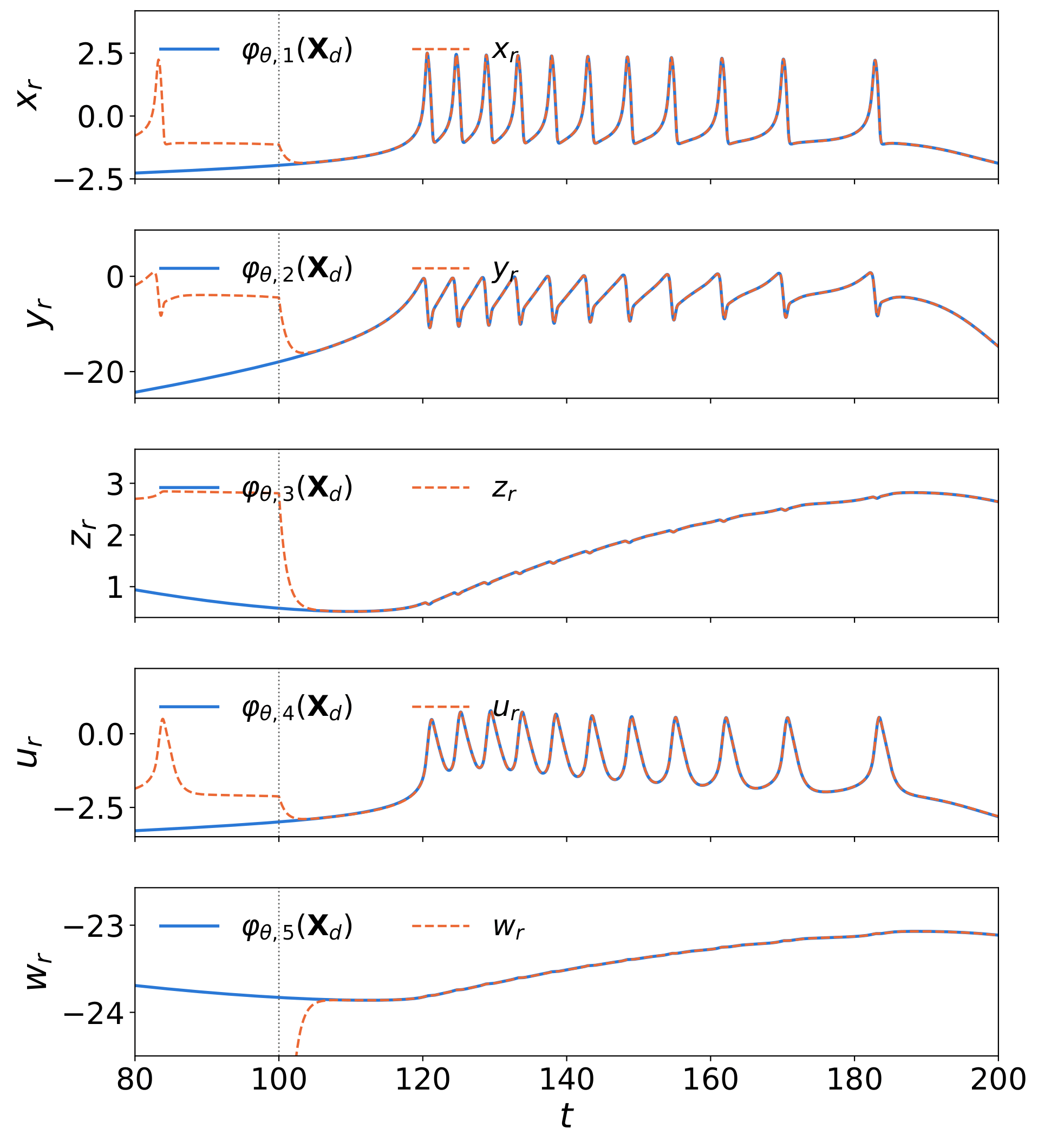}
\caption{Response components and learned targets $\varphi_{\theta,i}(\mathbf{X}_d)$ under data-driven control. The black vertical dotted line marks control activation at $t=100$.}
\label{fig:ml_closed_loop}
\end{figure}

\begin{figure}
\centering
\includegraphics[width=\columnwidth]{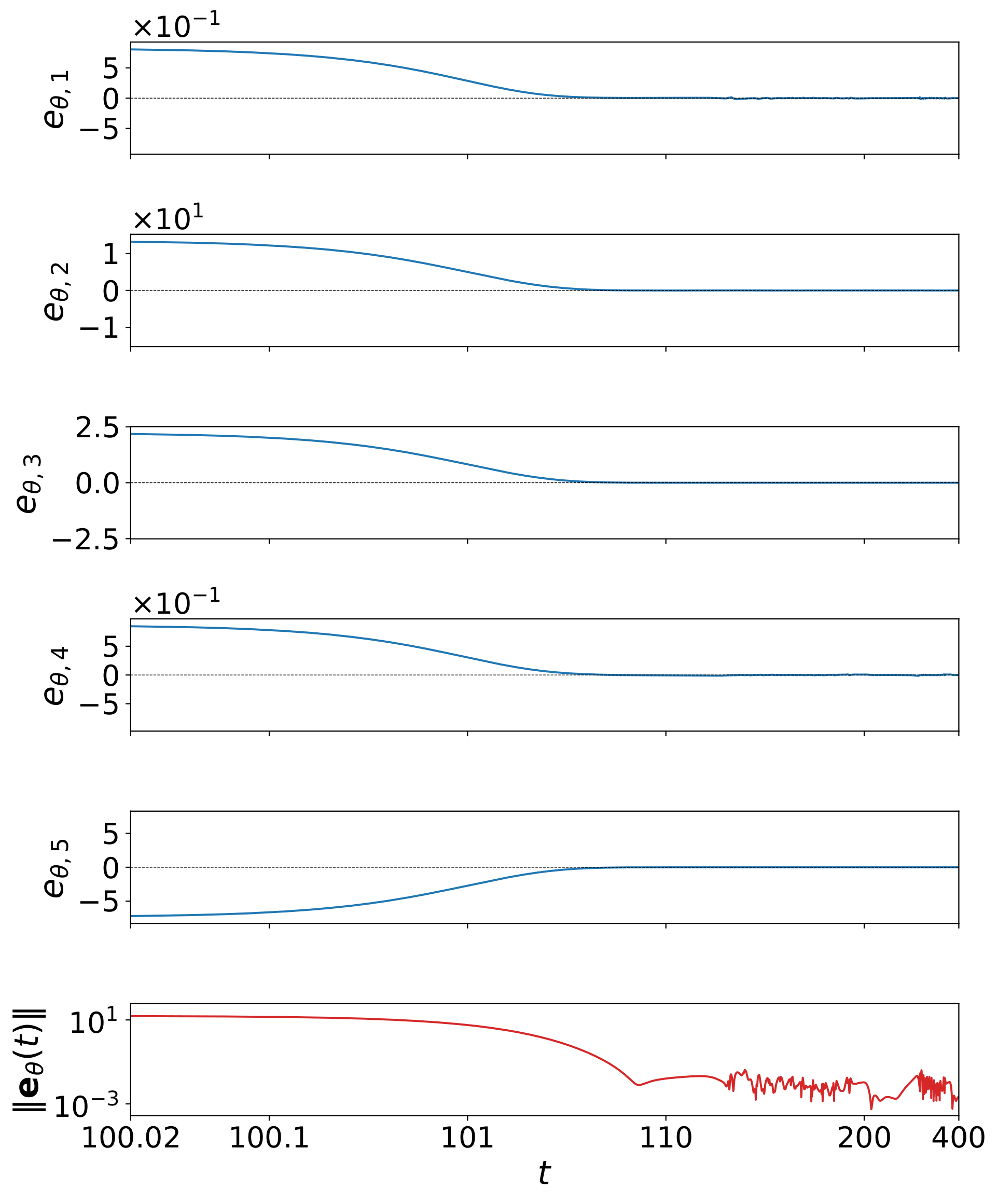}
\caption{Synchronization error components and their Euclidean norm under data-driven control, shown after activation at $t=100$ on logarithmic time axes.}
\label{fig:ml_sync_errors}
\end{figure}

\section{Summary and conclusion}
\label{sec:conclusion}
We have investigated increased-order generalized synchronization (IOGS) between a four-dimensional memristive Hindmarsh--Rose drive and a five-dimensional response containing an additional slow adaptation variable. The dimensional mismatch makes the specification of a synchronization relationship an integral part of the control problem: the target map must assign a dynamical correspondence to every response component, including the variable without a counterpart in the drive. We addressed this problem through two complementary constructions, one based on a prescribed nonlinear map and adaptive Lyapunov design, and the other on an affine map learned from the consistency of data-driven predictors. These constructions connect the choice of synchronization manifold to two distinct objectives: retaining an error structure amenable to stability analysis and reducing the dynamical mismatch that the control must compensate.

In the analytical formulation, the constraint $\varphi_4=u$ preserves the memristive switching difference explicitly in the closed-loop error dynamics. The adaptive laws cancel the parameter-error cross terms in the derivative of a quadratic Lyapunov function, while a case-wise analysis characterizes the remaining memristive contribution. Under the bounded-drive and global-existence assumptions of \cref{thm:adaptive_asymptotic}, uniform strict dissipativity of this contribution yields asymptotic convergence of the synchronization error and boundedness of the parameter-estimation errors. Without the strict dissipativity condition, \cref{thm:adaptive_practical_lyapunov} provides a switching-independent estimate of the form $\dot V\leq-\alpha\|e\|^2+C$, identifying an explicit radius $\delta=\sqrt{C/\alpha}$ outside which $V$ decreases. This radius characterizes a Lyapunov decrease threshold; the estimate does not, by itself, establish ultimate boundedness or practical synchronization in the asymptotic sense.

The numerical results delineate the distinction between these sufficient conditions and the observed dynamics. For the investigated trajectory at $m=0.5$, the memristive contribution is nonpositive but vanishes during visits to the central switching region $\mathcal C_5$, even when the corresponding error is nonzero. The uniform strict dissipativity hypothesis is therefore not satisfied, although all synchronization-error components decay toward zero over the simulated interval. The simulations thus indicate convergence beyond the conditions covered by the asymptotic theorem. Accurate recovery of the nominal parameters is also observed, but a general parameter-identification result would require additional excitation conditions that have not been established here.

The data-driven formulation addresses the same dimensional mismatch without supplying the governing equations to the learning or control algorithms. Memoryless next generation reservoir computing predictors are trained on uncoupled full-state recordings, and an affine synchronization map is learned by promoting agreement between mapping and one-step dynamical evolution. Variance and covariance penalties discourage collapsed target relationships without prescribing a coordinate correspondence. The resulting map approximately recovers the linear part of the reference correspondence while adjusting the offsets and the target associated with the additional response variable. A sampled-data feedback law constructed from the learned predictors then drives the response toward this target relationship, with the response parameters held fixed.

Across the eight initial-condition pairs considered, the learned map reduces the median control effort by approximately a factor of eight relative to the reference map within the same data-driven framework, while also reducing the median tracking error. These findings support the use of dynamical consistency as a criterion for selecting synchronization targets that require less corrective forcing. The remaining tracking error reflects limitations of both the learned predictors and the approximate sampled-data control model. Replacing the predictors with numerical flow maps obtained from the governing equations substantially reduces the residual, demonstrating a significant contribution from prediction error. The observed tracking performance nevertheless remains a numerical result rather than a proof of asymptotic synchronization or closed-loop stability.

The two approaches establish complementary roles for synchronization-map design. The analytical construction retains the memristive error structure and yields conditional convergence guarantees, whereas the learned construction seeks compatibility with the uncoupled dynamics and demonstrates reduced control effort. Their assumptions and target maps differ, so the numerical results do not constitute a direct performance comparison between the analytical and data-driven controllers. In particular, the data-driven implementation assumes full-state measurements and additive actuation of all response components, and its unconstrained affine map does not enforce the memristive-coordinate identity used in the analytical proof.

Several mathematical questions follow naturally from these results. A convergence analysis based on dissipation accumulated along trajectories could potentially weaken the uniform strict dissipativity requirement and account for the observed synchronization during intermittent visits to $\mathcal C_5$. For the data-driven scheme, stability estimates incorporating predictor errors and finite sampling would clarify how model accuracy and sampling rate determine the attainable tracking precision. Combining learned synchronization maps with explicit structural constraints offers a further direction for relating dynamical compatibility to rigorous convergence conditions. Extensions to partial observations and restricted actuation would subsequently help determine how these constructions perform under more limited access to the neuronal states.

\backmatter

\bmhead{Acknowledgements}
We thank the Erlangen National High Performance Computing Center
(NHR@FAU) for providing the computational resources used in this work.

\bmhead{Funding}
This work was funded by the Department of Data Science (DDS), Friedrich-Alexander-Universit\"at Erlangen-Nürnberg, Germany, and the Deutsche Forschungsgemeinschaft (DFG, German Research Foundation) via the grant YA 764/1-1 to M.E.Y—Project No. 456989199.

\bmhead{Data availability}
The code supporting the findings of this study is publicly available at \url{https://github.com/dennismstfc/iogs-analytical-vs-ml}. All data and figures can be regenerated with the scripts provided there.

\section*{Declarations}

\bmhead{Competing interests}
The authors have no relevant financial or non-financial interests to disclose.

\begin{appendices}

\section{Offline learning and online control for IOGS}
\label{appendix}
Algorithm~\ref{alg:data_driven_iogs} summarizes the offline learning
and online control stages described in
Sections~\ref{subsec:ngrc}--\ref{subsec:coupling}.
Entire recordings are assigned to disjoint training, evaluation,
and test sets; test data are reserved for final assessment.
Each predictor search uses 30 trials, including 8 initial random
proposals. Map optimization uses L-BFGS-B without parameter bounds,
with the exact gradient, relative objective-reduction tolerance
$10^{-10}$, and gradient infinity-norm tolerance $10^{-7}$.
The variance penalty is smoothed as specified in
Section~\ref{subsec:map}.
Control is activated at $t_{\mathrm{on}}=100$ in the reported
experiments.

\begin{algorithm}
\caption{Offline learning and online sampled-data control for IOGS}
\label{alg:data_driven_iogs}
\DontPrintSemicolon

\KwInput{Uncoupled training and evaluation recordings sampled at
$\Delta t$; search ranges for $(q,\nu)$; map sample size $N$;
weights $\lambda_{\mathrm{var}},\lambda_{\mathrm{cov}}$;
feedback rate $\kappa>0$; activation time $t_{\mathrm{on}}$.}
\KwOutput{One-step predictors $P_d,P_r$; affine synchronization
map $\varphi_\theta$; applied control sequence $\mathbf{v}(t)$.}

\BlankLine
\tcp{Stage 1: learn the uncoupled predictors}
\ForEach{$s\in\{d,r\}$}{
  Compute training mean $\bar{\mathbf{X}}_s$ and scaling
  $\mathbf{D}_s$; standardize training and evaluation data
  using these statistics\;
  \For{$j\leftarrow1$ \KwTo $30$}{
    Propose $(q_j,\nu_j)$ randomly for $j\le8$,
    otherwise by Bayesian optimization\;
    Construct features and increment-rate targets;
    fit $\mathbf{W}_j$ using Eqs.
    \eqref{eq:ngrc_features} and \eqref{eq:ngrc_ridge}\;
    Evaluate rollout error $J_j$ on evaluation recordings
    using Eq. \eqref{eq:ngrc_rmse}\;
  }
  Retain $(q_j,\mathbf{W}_j)$ with the smallest $J_j$
  and define the corresponding increment-rate model $\hat F_s$\;
  $P_s(\mathbf{X})\leftarrow
  \mathbf{X}+\Delta t\,\hat F_s(\mathbf{X})$\;
}

\BlankLine
\tcp{Stage 2: learn the synchronization map}
Sample $N$ drive training states; compute their standardized
covariance $\mathbf{C}_d$ and the covariance $\mathbf{R}$
of the first four standardized response training components\;
Initialize $\theta=(\mathbf{A},\mathbf{b})
\leftarrow(\mathbf{0},\mathbf{0})$\;
With $P_d$ and $P_r$ fixed, optimize $L(\theta)$ in
Eq.~\eqref{eq:map_loss} by L-BFGS-B, using the two-path
residuals in Eq.~\eqref{eq:two_paths} and the stated
stopping criteria\;

\BlankLine
\tcp{Stage 3: control the response online}
$\zeta\leftarrow\exp(-\kappa\Delta t)$\;
\ForEach{sampling time $t$}{
  $\mathbf{v}(t)\leftarrow\mathbf{0}$\;
  \If{$t\ge t_{\mathrm{on}}$}{
    Measure $\mathbf{X}_d(t)$ and $\mathbf{X}_r(t)$\;
    $\mathbf{e}\leftarrow
    \mathbf{X}_r(t)-\varphi_\theta(\mathbf{X}_d(t))$\;
    $\mathbf{v}(t)\leftarrow
    \bigl[
      \varphi_\theta(P_d(\mathbf{X}_d(t)))
      -P_r(\mathbf{X}_r(t))
      +\zeta\mathbf{e}
    \bigr]/\Delta t$\;
  }
  Hold $\mathbf{v}(t)$ constant on $[t,t+\Delta t)$
  and apply it additively to the response\;
}
\end{algorithm}

\end{appendices}


\end{document}